\documentclass[11pt,a4paper,reqno]{amsart}
\usepackage[foot]{amsaddr}
\usepackage{amssymb}
\usepackage{amscd,bm}
\usepackage[pdftex,pdfpagelabels]{hyperref}
\usepackage{enumerate}
\usepackage{algorithm}
\usepackage{scalerel}
\usepackage{algpseudocode}
\usepackage{graphicx}
\usepackage{mathtools}
\newcommand{\pr}{\textnormal{Pr}}

\usepackage{siunitx}
\usepackage{tikz-cd}
\usepackage{bm}
\newcommand{\vone}{\mathbf{1}}
\numberwithin{equation}{section}

\newcommand{\omegam}{\Omega_{\scaleto{\ref{lem:delocalization}}{3pt}}}

\usepackage[obeyFinal,textsize=footnotesize]{todonotes}

\newcommand{\NN}{\mathsf{N}}
\newcommand{\sgn}{\textnormal{sgn}}

\usepackage{mathtools}
\usepackage[tableposition=top]{caption}
\usepackage{booktabs,dcolumn}

\DeclareFontFamily{OT1}{rsfs}{}
\DeclareFontShape{OT1}{rsfs}{n}{it}{<-> rsfs10}{}
\DeclareMathAlphabet{\mathscr}{OT1}{rsfs}{n}{it}

\theoremstyle{plain}

\newtheorem{theorem}{Theorem}[section]
\newtheorem{proposition}[theorem]{Proposition}
\newtheorem{prop}[theorem]{Proposition}

\newtheorem{lemma}[theorem]{Lemma}

\newtheorem{conjecture}[theorem]{Conjecture}

\newtheorem{claim}[theorem]{Claim}

\theoremstyle{definition}

\newtheorem{definition}[theorem]{Definition}
\newtheorem{remark}[theorem]{Remark}

\newtheorem{example}[theorem]{Example}

\newcommand\E{\mathbb{E}}
\newcommand\R{\mathbb{R}}

\newcommand\N{\mathbb{N}}

\newcommand\pphi{\bm{\varphi}}

\begin{document}
\title[Nodal Decompositions of a Symmetric Matrix]{Nodal Decompositions of a Symmetric Matrix}
\author{Theo McKenzie}
\address{Department of Mathematics, Stanford University, Stanford, CA, USA.}
\email{theom@stanford.edu}
\thanks{TM is supported by NSF Grant DMS-2212881}
\author{John Urschel}
\address{Department of Mathematics, Massachusetts Institute of Technology, Cambridge, MA, USA.}
\email{urschel@mit.edu}

\begin{abstract}
  Analyzing nodal domains is a way to discern the structure of eigenvectors of operators on a graph. We give a new definition extending the concept of nodal domains to arbitrary signed graphs, and therefore to arbitrary symmetric matrices. We show that for an arbitrary symmetric matrix, a positive fraction of eigenbases satisfy a generalized version of known nodal bounds for un-signed (that is classical) graphs. We do this through an explicit decomposition. Moreover, we show that with high probability, the number of nodal domains of a bulk eigenvector of the adjacency matrix of a signed Erd\H os-R\'enyi graph is $\Omega(n/\log n)$ and $o(n)$. 
\end{abstract}

\maketitle

\section{Introduction}\label{sec:intro}
Courant's nodal domain theorem states that the zero-level set (i.e., the set of points where the eigenfunction equals zero) of the $k^{th}$ lowest energy eigenfunction of a Laplacian on a smooth bounded domain in $\mathbb{R}^d$ with Dirichlet boundary conditions  divides the domain into at most $k$ subdomains (see \cite{courant1923allgemeiner}, and his text co-authored with Hilbert in the following year \cite{courant1924methoden}). The zero-level set is commonly referred to as the nodal set, the resulting subdomains are referred to as the nodal domains, and the number of subdomains is referred to as the nodal count. Results of this type have been of great interest in spectral geometry and mathematical physics (see, e.g. \cite{zelditch2017eigenfunctions}), with refinements in dimension two (e.g., Pleijel's nodal domain theorem) \cite{bourgain2013pleijel,pleijel1956remarks,polterovich2009pleijel}, and extensions to $p$-Laplacians \cite{cuesta2000nodal,drabek2002generalization}, Riemannian manifolds \cite{peetre1957generalization}, and domains with low regularity assumptions \cite{alessandrini1998courant}, among many others.

 Courant's theorem, and nodal domains in general, has also been studied in the discrete setting of graphs. This setting poses a number of unique challenges, as eigenvectors may vanish (e.g., equal zero) at some entries, while in Courant's setting the nodal set is of measure zero \cite{cheng1976eigenfunctions}. In this setting, a nodal domain of an eigenvector $\bm{\varphi}$ of the generalized Laplacian of a graph
is a maximal connected component on which the eigenvector entries do not change sign, e.g. ${\pphi(i) \pphi(j) > 0}$ for all $i \ne j$ in the domain.

 The earliest known result in this setting is due to Gantmacher and Krein, who studied the sign properties of eigenvectors of generalized Laplacians of the path graph and proved a tight estimate for the nodal count in this setting (see \cite{gantmacher2002oscillation} for a revised English edition of the original 1950 Russian text). 
 Fiedler's tree theorem proves exact nodal count estimates for trees, generalizing the work of Gantmacher and Krein. Namely he showed that the number of nodal domains of a non-vanishing eigenvector of a symmetric, acyclic, irreducible matrix is exactly the index of the corresponding eigenvalue indexed in increasing order \cite{fiedler1975eigenvectors} (in fact, both Gantmacher and Krein's result and Fiedler's result extend to the signed case, described below). These results can be thought of as a discrete version of Sturm's oscillation theorem for ordinary differential equations \cite{sturm2009memoire1,sturm2009memoire2}, of which Courant's theorem is a generalization. 

For discrete generalized Laplacians, the nodal count of a non-vanishing eigenvector corresponding to the $k^{th}$ eigenvalue is at most $k$. However, when an eigenvector vanishes on some vertices, complications arise, as the vertex sets of nodal domains no longer forms a partition of the vertex set. In this setting, there are a number of competing versions of nodal domains and nodal theorems. Most notably, many authors have considered the concept of weak and strong nodal domains: a strong nodal domain of an eigenvector $\bm{\varphi}$ of a symmetric matrix $M$ is simply a nodal domain as defined above, i.e., a maximally connected induced subgraph for which {$\bm{\varphi}(i) \bm{\varphi}(j)>0$}, and a weak nodal domain is a maximally connected induced subgraph for which {$ \bm{\varphi}(i) \bm{\varphi}(j)\ge0$, see Figure \ref{fig:examples}}.

Davies, Gladwell, Leydold, and Stadler proved a weak and strong nodal count theorem for generalized Laplacians {(e.g., unsigned graphs)}: given an eigenpair $(\lambda,\bm{\varphi})$ of an irreducible generalized Laplacian $M$, the weak and strong nodal count of $\bm{\varphi}$ are at most $k$ and $k+r-1$, respectively, where $k$ and $r$ are the index (in increasing order) and multiplicity of $\lambda$, respectively \cite{biyikoglu2007laplacian,BRIANDAVIES200151,davies2000discrete}. There are a number of other proofs of various versions of this statement \cite{de1993multiplicites,duval1999perron,friedman1993some,powers1988graph,van1996topological}; see \cite[Sec. 2]{BRIANDAVIES200151} for a discussion of the results (and the correctness of some of the statements and associated proofs) in these works. {In \cite{urschel2018nodal},} the second author of the current paper proved a decomposition version of the Davies-Gladwell-Leydold-Stadler theorem, that for the $k^{th}$ eigenvalue of a symmetric generalized Laplacian of a discrete graph, a positive proportion of eigenvectors $\bm{\varphi}$ in the corresponding eigenspace can be decomposed into at most $k$ signed nodal domains (i.e., there exists a signing $\bm{\varepsilon}$  {satisfying $\bm{\varepsilon}(i) = \sgn(\bm{\varphi}(i))$ whenever $\bm{\varphi}(i) \ne 0$} with classical nodal count at most $k$).

In addition to discrete versions of Courant's nodal theorem, measuring the gap between the actual nodal count and Courant's bound has also been of interest. Most notably, Berkolaiko showed that a non-vanishing eigenvector of an irreducible generalized Laplacian has nodal count at least $k-\nu$, where $k$ is the index of the corresponding eigenvalue and $\nu$ is the cyclomatic number of the associated graph \cite{berkolaiko2008lower}\footnote{Berkolaiko stated a slightly weaker version of the aforementioned result, but the associated proof also proves the aforementioned version.}. {This result is also a corollary of the more general Berkolaiko-Colin de
Verdi\`ere theorem, in which the nodal surplus is expressed as a Morse index of a function of $\nu$ variables  \cite{berkolaiko2013nodal,colin2013magnetic}. Berkolaiko's orginal result} was later strengthened by Xu and S.T. Yau, producing a lower bound for an arbitrary eigenvector $\bm{\varphi}$ of $k+r-1-\nu - |i_0(\bm{\varphi})|$, where $k$ and $r$ are the index and multiplicity of the corresponding eigenvalue, $\nu$ is the cyclomatic number of the associated graph, and $i_0(\bm{\varphi}) = \{i \, | \, \bm{\varphi}(i) = 0 \}$ \cite{xu2012nodal}. The study of nodal surplus {(stemming from \cite{berkolaiko2013nodal,colin2013magnetic})} is an active area of research, see \cite{alon2020quantum,alon2018nodal,alon2022universality,alon2022morse} for details.

In this work, we consider nodal count theorems for the generalized Laplacian of arbitrary signed graphs, e.g., arbitrary symmetric matrices. Such estimates are important for a number of reasons. First, when dealing with a finite element approximation of an elliptic operator, the resulting stiffness matrix is not always a generalized Laplacian and Courant's bound may fail to hold (e.g., a 2D triangularization with some obtuse angles); see, for instance, \cite{gladwell2002courant} for details. More generally, nodal domains give us an idea of relationship between the structure of a matrix and that of its eigenvector, and signed graphs are used frequently in practice (e.g., Ising models, correlation clustering, etc). {For instance, signed graphs played a role in the recent breakthroughs regarding equiangular lines with a fixed angle \cite{jiang2021equiangular, jiang2023spherical}, the sensitivity conjecture \cite{huang2019induced}, and Bilu-Linial lifts to construct infinite families of Ramanujan graphs \cite{bilu2006lifts, marcus2013interlacing}.}

 However, as the setting of discrete Laplacians added barriers (in the form of vanishing vertices) to a direct version of the classical Courant nodal domain theorem, the extension from generalized Laplacian matrices to arbitrary symmetric matrices brings unique challenges, and ambiguity as to what constitutes a nodal domain in this setting. Path nodal domains for symmetric matrices have been studied by Mohammadian \cite{mohammadian2016graphs} and, recently, by Ge and Liu \cite{ge2023symmetric}. In this model, we count the number of connected components of the graph induced on ``good'' edges, namely edges for which the product of eigenvector entries of vertices in the edge respects the sign of the edge. Namely, given a symmetric irreducible matrix $M \in \mathbb{R}^{n \times n}$ and an eigenvector $\bm{\varphi}$ corresponding to an eigenvalue of index $k$ and multiplicity $r$, let 
$G^<_{\bm{\varphi}} = (\{1,...,n\},E^<_{\bm{\varphi}})$ and $G^\le_{\bm{\varphi}} = (\{1,...,n\} \backslash i_0(\bm{\varphi}),E^\le_{\bm{\varphi}})$, where
$E^<_{\bm{\varphi}} = \{(i,j) \in E(G) \, | \,  M_{ij} \bm{\varphi}(i) \bm{\varphi}(j) <0 \}$ and $E^\le_{\bm{\varphi}} = \{(i,j) \in E(G) \, | \,  M_{ij} \bm{\varphi}(i) \bm{\varphi}(j) \le 0 \}$ (with $G^\ge_{\bm{\varphi}}$ and $G^>_{\bm{\varphi}}$ defined analogously). Let $\kappa(\cdot)$ be the number of connected components of a graph. Mohammadian proved that $\kappa(G^\le_{\bm{\varphi}}) \le k$, $\kappa(G^<_{\bm{\varphi}}) \le k+(r-1)$, and that, if $i_0(\bm{\varphi}) = \emptyset$, then $\kappa(G^<_{\bm{\varphi}}) \ge k - \nu$ \cite{mohammadian2016graphs}. Ge and Liu expanded upon the analysis of Mohammadian by proving lower bounds for $\kappa(G^<_{\bm{\varphi}})$ {involving further parameters that depend on the eigenvector $\bm{\varphi}$}, and showing that $\kappa(G^<_{\bm{\varphi}})\le k$ for any eigenvector of minimal support \cite{ge2023symmetric}. In addition, they produced an upper bound on a formulation of weak nodal domain slightly different from $G^\le_{\bm{\varphi}}$, and proved a number of estimates for acyclic matrices.

There is a peculiarity specific to the path nodal domain, in that it is possible that $i,j$ in the same nodal domain could satisfy $M_{ij}\pphi(i)\pphi(j)> 0$, whereas this is impossible in un-signed graphs. Namely, $i$ and $j$ could be identifiably negatively related, but still be a part of the same nodal domain, if there is a path of good edges from $i$ to $j$. Here, ``bad'' edges are treated equivalently to as if there were no edge at all.

{In order to incorporate the information from all edges into our decomposition, we take a different approach from that of previous authors. Instead of studying walks with a classical nodal-type property and ignoring bad edges (i.e. considering only $M_{ij} \bm{\varphi}(i) \bm{\varphi}(j) < 0$), we prove bounds for induced subgraphs (``nodal" subgraphs) with a nodal-type property, thus producing subsets of the domain for which the eigenvector does not change sign and the matrix, restricted to this subset, is a generalized Laplacian (both up to a sign transformation $M \rightarrow DMD$, $\bm{\varphi} \rightarrow D \bm{\varphi}$, for some involutory diagonal matrix). We consider bounds involving properties of the matrix itself, rather than quantities that depend on the specific choice of eigenvector (see, for instance, Remark \ref{remark:invariants}).}


In this work, we focus on the minimal size of decompositions of the domain into nodal subgraphs ({see Figure \ref{fig:examples}}). This is a stricter definition than that of path nodal domains; therefore, we expect more nodal domains in this case. In what follows, we prove upper bounds regarding nodal decompositions of a symmetric matrix. In particular, we prove a natural analogue of Courant's nodal domain theorem for any matrix and eigenvector pair that depends only on the energy level of the eigenvector and how ``close" to a generalized Laplacian the given matrix is (Theorem \ref{thm:basis_bound}).

\begin{figure}
{\centering
\includegraphics[width = 6.5 in]{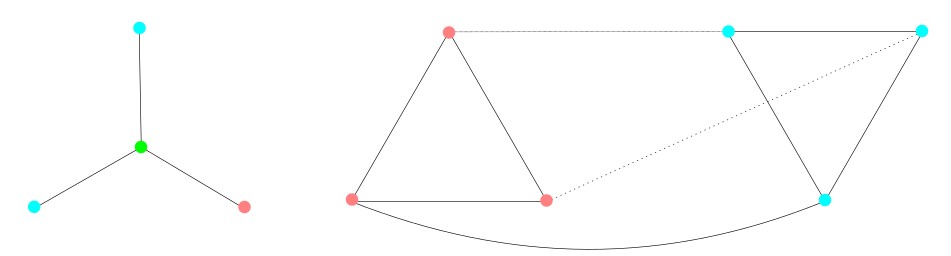}
\caption{{In both (unweighted) graphs, solid lines represent positive edges and dashed lines represent negative edges. Moreover, red vertices represent positive entries in the eigenvector, blue vertices are negative, and green is 0. For the given eigenvector on the graph on the left, there are 2 weak nodal domains but 3 strong nodal domains. For the eigenvector $\pphi$ on the graph on the right, there is 1 path nodal domain, but $\mathsf N(\pphi)=2$.}}\label{fig:examples}}
\end{figure}

One large appeal of our formulation is in deducing the structure of eigenvectors of random graphs. Studying the structure of eigenvectors of random graphs is a well known problem with applications in both computer science and mathematical physics, see, for example,  \cite{kottos1997quantum,pothen1990partitioning}. On dense Erd\H{o}s-R\'enyi graphs, eigenvectors that do not correspond to the highest eigenvalue have exactly two nodal domains \cite{arora2011eigenvectors,dekel2011eigenvectors}, which are approximately the same size \cite{huang2020size}, and all vertices are on the boundary of their nodal domains \cite{rudelson2017delocalization}. This follows the general notion of \emph{quantum ergodicity}, that roughly, the distribution of entries in an eigenvector of these graphs should be close to a joint Gaussian distribution with individual entries close to independent \cite{bourgade2017eigenvector}. 

All of these results show that Erd\H{o}s-R\'enyi graphs have ``trivial'' eigenvector structure, in that the graph is too dense for these nodal domain statistics to detect different structure within the graph (this is not the case for random regular graphs, for which eigenvectors of the most negative eigenvalues have many nodal domains \cite{ganguly2023many}). In our setting, we consider an Erd\H os-R\'enyi \emph{signed} graph, denoted by $G(n,p,q)$ for $n\in \N$, $0<p,q<1$. Here, we randomly sample an $n$ vertex graph, where each of the $\binom n2$ possible edges has a $p$ probability of being a positive edge, a $q$ probability of being a negative edge, and $1-p-q$ of not existing. {For an overview of the spectral theory of signed random graphs see \cite{gallier2016spectral}. Signed Erd\H{o}s-R\'enyi graphs are relevant in spin glass models, where positive and negative edges correspond to ferromagnetic/antiferromagnetic bonds in the discrete Hamiltonian,  \cite{dean2001tapping}. Moreover, the spectral theory signed Erd\H os-R\'enyi graphs is relevant for analyzing real world community detection problems \cite{ mercado2019spectral,tang2016survey}.}

Using the same argument as \cite{dekel2011eigenvectors}, with high probability, every eigenvector $\pphi$ has $\kappa(G^{<}_{\bm{\varphi}})=\kappa(G^{>}_{\bm{\varphi}})=1$ (we prove this in Appendix \ref{sec:averagepath}). However, as we show in Section \ref{sec:average}, our stronger notion of a nodal domain shows the nontrivial structure of the eigenvectors of a $G(n,p,q)$, as the number of nodal domains scales sublinearly. Specifically, we show that with high probability, there are $o(n)$ nodal domains in this signed random graph. A simple lower bound of $\Omega(n/\log_{\frac{1}{1-(p\vee q)}}n)$ is given in Proposition \ref{prop:avglower}, which follows from the size of the maximum clique in a graph. These results are also towards a program proposed by Linial {(described in \cite{huang2020size})} studying the geometry of nodal domains on graphs. 

\subsection{Definitions, Notation, and Results}
 Let $M$ be an $n \times n$ symmetric, irreducible matrix with eigenpair $(\lambda, \bm{\varphi})$. We denote the index of $\lambda$ (the number of eigenvalues strictly less than $\lambda$, plus one) by $k$, and the multiplicity by $r$. We can associate with $M$ a signed graph $G = ([n],E,\sigma)$, where $[n]:=\{1,...,n\}$, $E = \{(i,j) \, | \, M_{ij} \ne 0, i \ne j \}$, and $\sigma:E \rightarrow \{\pm 1\}$ is defined by $\sigma_{ij} = -\sgn(M_{ij})$. Let $G[S]$ be the induced subgraph of $G$ on vertex set $S \subset [n]$, with corresponding edge set $E[S]$. We denote by $\nu$ the cyclomatic number (dimension of the unsigned cycle space) of $G$. Given a subset of vertices $S$, we denote by $M_S$ the $|S|\times |S|$ principal submatrix of $M$ corresponding to $S$. Similarly, we write $\bm{\varphi}(S)$ as the subvector of $\bm{\varphi}$ corresponding to $S$. We denote by $\overrightarrow{1_{S}}$ the vector that equals one on $S$ and zero elsewhere; its dimension is always clear from context. For a vertex $v\in V$, we write $\pphi(v)=\pphi(\{v\})$ and $\overrightarrow{1_{v}} = \overrightarrow{1_{\{v\}}}$.  For two vectors $x,y$ of the same length, $x\circ y$ denotes the entrywise product of $x$ and $y$. Let $p \vee q = \max\{p,q\}$ and $p \wedge q = \min \{p,q\}$.

A key ingredient in the worst-case analysis that follows is a notion of how ``far" a signed graph is from being equivalent to an all positive graph, e.g., how far a matrix is from being a generalized Laplacian. We make use of the notion of frustrated edges and frustration index.
\begin{definition}
Given a signed graph $G = ([n],E,\sigma)$ and a \emph{state} $\bm{\varepsilon} \in \{\pm 1\}^n$, an edge $\{i,j\} \in E$ is said to be \emph{frustrated} if $\sigma_{i,j} \bm{\varepsilon}(i) \bm{\varepsilon}(j)<0$. The \emph{frustration index} $\mathsf{f}$ of $G$ is the minimum number of frustrated edges over all states $\bm{\varepsilon} \in \{\pm 1\}^n$.
\end{definition}
 Computing the frustration index of a signed graph is $\mathsf{NP}$-hard, via a reduction from the max-cut problem; {see \cite{huffner2010separator} for a through discussion of the complexity of computing the frustration index.} In what follows, we often assume that the number of positive off-diagonal pairs of $M$ is exactly $\mathsf{f}$; this can be done by performing an involutory diagonal transformation $DMD$ for a diagonal matrix $D$ corresponding to a state $\bm{\varepsilon} \in \{\pm 1\}^n$ that achieves the frustration index. {When $G$ is a forest, corresponding to the sparsity structure of an acyclic matrix, the frustration index is exactly zero, as every symmetric acyclic matrix can be transformed to a generalized Laplacian via an involutory diagonal matrix.} Given a vector $\bm{x}$, we denote by $i_0(\bm{x})$ the set of indices where $\bm{x}$ equals zero, i.e., $i_0(\bm{x}) := \{ j \, | \,  \bm{x}(j) = 0 \}$.
\begin{definition}
Given a symmetric matrix $M \in \mathbb{R}^{n \times n}$ and a non-vanishing vector $\bm{x} \in \mathbb{R}^n$ (i.e., $i_0(\bm{x}) = \emptyset$), we denote by $\mathsf{N}(\bm{x})$ the minimal quantity $s$ for which there exists a partition $[n] = \sqcup_{\ell=1}^s V_\ell$ with $G[V_\ell]$ connected and $M_{ij} \bm{x}(i) \bm{x}(j) <0$ for all $\{i,j\} \in E[V_\ell]$, $\ell = 1,...,s$, i.e., the minimal decomposition of the domain into nodal subgraphs.
\end{definition}
{Computing $\mathsf{N}(\bm{\varphi})$ is $\mathsf{NP}$-hard (as is often the case for problems involving signed graphs), via a reduction from the clique-cover problem (one of Karp's original $21$ $\mathsf{NP}$-complete problems \cite{karp2010reducibility}). We provide a brief sketch of a reduction. Given an instance $H = ([n],E)$ of the clique-cover problem, consider the signed clique $G = ([n],E,\sigma)$, where $\sigma_{ij} = +1$ if $(i,j) \in E(H)$ and $\sigma_{ij}=-1$ if $(i,j) \not\in E(H)$. Computing the smallest $s$ such that there exists a partition $[n] = \sqcup_{\ell=1}^s V_\ell$ with $G[V_\ell]$ connected and $\sigma_{ij}=+1$ for all $(i,j) \in E[V_\ell]$, $\ell = 1,...,s$, is exactly the smallest $s$ such that the vertices of $H$ can be partitioned into $s$ parts, each of which is a clique. Taking any matrix with sparsity structure $G$ and considering a vector $\bm{\varphi}$ with all entries positive implies that computing $\mathsf{N}(\bm{\varphi})$ is also $\mathsf{NP}$-hard.} This hardness holds not just for adversarial $\bm{\varphi}$, but for eigenvectors $\bm{\varphi}$, as, for any signed graph $G = ([n],E,\sigma)$, one can always produce a matrix $M$ with sparsity structure $G$ and an eigenvector with constant sign. We provide a short proof of the bounds
\begin{equation} \label{ineq:generic1}
    k +(r-1) - \mathsf{\nu}  \le \mathsf{N}(\bm{\varphi}) \le k + \mathsf{f}
\end{equation}
for non-vanishing eigenvectors
in Section \ref{sec:non-vanishing} (Proposition \ref{prop:generic_bound}).
This follows quickly from Mohammadian's aforementioned results and an analysis of frustrated edges, but we provide a short proof of independent interest. Let $\mathsf{N}^s(\bm{\varphi})$ be the minimal nodal decomposition of $\bm{\varphi}$ restricted to non-vanishing entries (i.e., $\mathsf{N}^s(\bm{\varphi}) = \mathsf{N}(\bm{\varphi} \vert_{\bm{\varphi}(i) \ne 0})$). We prove that for any symmetric matrix $M$ there exists an orthonormal eigenbasis $\bm{\varphi}_1,...,\bm{\varphi}_n$ ordered by increasing energy, such that {$\mathsf{N}^s(\bm{\varphi}_k)\le k+ \mathsf{f}$} for all $k$ (Proposition \ref{prop:strong_basis}), extending a result of Gladwell and Zhu for generalized Laplacians \cite{gladwell2002courant}. However, the restriction to non-vanishing entries leads to a number of limitations in the above result. For instance, the eigenbases satisfying the proposition statement may be of measure zero. For this reason, we prove a more robust version of the above statement, focusing on nodal decompositions of the entire domain. In particular, in Section \ref{sec:eigenbasis}, we prove that there exists a subset of orthonormal eigenbases, of positive measure, such that for every basis $\bm{\varphi}_1,...,\bm{\varphi}_n$ in this subset ordered by increasing energy, there exist corresponding signings $\bm{\varepsilon}_1,...,\bm{\varepsilon}_n \in \{\pm 1\}^n$ satisfying {$\bm{\varepsilon}_k(i) = \sgn(\bm{\varphi}_k(i))$ whenever $\bm{\varphi}_k(i) \ne 0$} and $\mathsf{N}(\bm{\varepsilon}_k) \le k+\mathsf{f}$ for all $k$.

\begin{theorem}\label{thm:basis_bound}
Let $M$ be an $n \times n$ symmetric, irreducible matrix, and let $\mathcal{B}$ be the set of corresponding orthonormal eigenbases of $\mathbb{R}^n$ {ordered by increasing energy, i.e.,
$$ \mathcal{B} = \{ Q \in O(n) \, | \, M Q = Q \Lambda \text{ for diagonal } \Lambda \text{ with } \Lambda_{1,1} \le \Lambda_{2,2} \le ... \le \Lambda_{n,n} \}.$$}
Then there exists a subset $\Phi \subset \mathcal{B}$ of co-dimension zero such that, for every $(\bm{\varphi}_1,...,\bm{\varphi}_n) \in \Phi$, there exists signings $\bm{\varepsilon}_1,...,\bm{\varepsilon}_n \in \{\pm 1\}^n$ satisfying {$\bm{\varepsilon}_k(i) = \sgn(\bm{\varphi}_k(i))$ for all $\bm{\varphi}_k(i) \ne 0$, $i,k \in [n]$,} and
\begin{equation}\label{ineq:basis}
 \mathsf{N}(\bm{\varepsilon}_k) \le k+\mathsf{f}, \qquad \quad k =1,...,n,
\end{equation}
where $\mathsf{f}$ is the frustration index of the signed graph of $M$.
\end{theorem}
Informally, the above theorem tells us that if we choose an arbitrary orthonormal eigenbasis, there is a positive probability that there is a signing of vanishing entries of our eigenbasis so that the resulting vectors satisfy the upper bounds of Inequality \ref{ineq:generic1}. This can be viewed as a stronger version of ``weak nodal bounds," as vanishing vertices cannot be used to connect both positively and negatively signed vertices. However, we can only hope for such a result for a positive proportion of bases. The unsigned star graph is an illustrative example of this limitation, {see Figure \ref{fig:star}.}
\begin{figure}
{\centering
\includegraphics[width = 6.5 in]{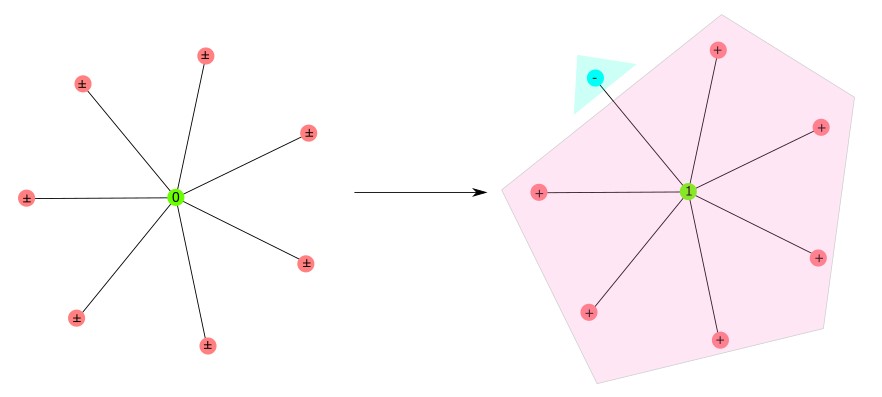}
\caption{{The graph Laplacian of a star on $n>3$ vertices has algebraic connectivity one with corresponding eigenspace $E_1 = \{ \bm{x} \, | \, \sum_{i = 2}^n \bm{x}(i) = 0, \, \bm{x}(1) = 0 \}$. In order for $\mathsf{N}(\bm{\varepsilon}_2) \le 2$, $\bm{\varphi}_2 \in E_1$ must have either exactly one positive entry or exactly one negative entry, a property only an exponentially small fraction (with respect to $n$) of the eigenvectors in $E_1$ satisfy in respect to the Haar measure over the eigenspace.}}\label{fig:star}}
\end{figure}

The positive probability portion of the above theorem is quite important; this property forces the resulting partitions in the basis to represent the ``dynamics" of the eigenspace, e.g., the theorem statement does not require the artificial vanishing of vertices (as in the proof of Proposition \ref{prop:strong_basis}). Our proof of Theorem \ref{thm:basis_bound} proceeds as follows:

\begin{enumerate}
\item
Characterize the structure of an eigenspace $E_\lambda$ with eigenvectors that simultaneously vanish, i.e., $i_0(\lambda):= \{ j \, | \, \bm{\varphi}(j) = 0 \text{ for all } \bm{\varphi} \in E_\lambda\} \ne \emptyset$, and parameterize $E_\lambda$ using eigenvectors of the connected components of $M$ restricted to $[n] \backslash i_0(\lambda)$.

\item
Algorithmically define the sign vector $\bm{\varepsilon}$ associated with any eigenvector $\bm{\varphi}$ that is non-vanishing on $[n] \backslash i_0(\lambda)$ and restrict the elements of the orthonormal basis $\bm{\varphi}_1,...,\bm{\varphi}_r$ of $E_\lambda$ to certain half-spaces so that $\mathsf{N}(\bm{\varepsilon}_s) \le k+(s-1) +\mathsf{f}$.

\item
Prove that a positive proportion of orthonormal bases of $E_\lambda$ are non-vanishing on $[n] \backslash i_0(\lambda)$ and satisfy the half-space conditions of {Step (2)}.
\end{enumerate}

{It is possible to prove an analogue of Theorem \ref{thm:basis_bound} for the lower bound $k-\nu$ instead of the upper bound $k + \mathsf{f}$; the analysis is similar to (but simpler than) that of the stated theorem. However, the more interesting statement regarding bases satisfying the lower and upper bounds simultaneously (i.e., $k- \nu \le \mathsf{N}(\bm{\varepsilon}_k) \le k + \mathsf{f}$) does not follow directly from the techniques presented, and may require a more detailed analysis. This possible extension is left to the interested reader.}

In Section \ref{sec:average}, we analyze the nodal count of the adjacency operator of the Erd\H{o}s-R\'enyi signed graph. A lower bound is given by the combinatorial properties of a $G(n,p,q)$, that a nodal domain is similar to a clique, and that all cliques in a random graph are of size $O(\log n)$. Therefore there are $\Omega(n/\log n)$ nodal domains {(Proposition \ref{prop:avglower})}. An upper bound proves to be a tougher challenge; however, we show the following in Section \ref{sec:average}.

{
\begin{theorem}\label{thm:mainavg}
For any $0<\epsilon,p,q<1$, there is some constant $\gamma>0$ such that for any index $i\in [\epsilon n,(1-\epsilon)n]$, with probability $1-O(n^{-\gamma})$, the $i$th eigenvector $\bm{\varphi}$ of the adjacency matrix of $G\sim G(n,p,q)$ has $\NN(\bm{\varphi})=o(n)$ nodal domains. 
\end{theorem}
}

We do this by showing that for any fixed $k$, we can partition almost the entire graph into nodal domains of size $k$. In order to show that there are many nodal domains in our graph of size $k$, we use quantum ergodicity \cite{huang2015bulk} to show that eigenvector statistics emulate those of random Gaussians. Therefore, fix $k$, and consider $\bm{\varphi}$ the $i$th eigenvector of $A$, where $A$ is the adjacency matrix of the graph $G\sim G(n,p,q)$. We proceed as follows.

\begin{enumerate}
\item
For a set of vertices $S$, create a function $f_s(A_S,\bm{\varphi}(S))$, {where $A_S$ is $A$ restricted to the set $S$}, on the induced subgraph on $S$ that confirms there is a nodal domain inside $S$. 

\item
Approximate $f$ with a finite degree polynomial $p_s$ of entries of the matrix and entries of the eigenvector. 
\item
Show that the matrix entries and eigenvector are close to independent in $p_s$. 

\item
Use quantum ergodicity to compare $\E(p_s(A_S,\bm{\varphi}))$ with $\E(p_s(A_S,\mathbf{g}))$, where $\mathbf{g}$ is a multivariate standard normal Gaussian.

\item
Show that with Gaussian inputs there are many nodal domains. 
\end{enumerate}
This general method was used in \cite{huang2020size} to show the two nodal domains of an unsigned Erd\H{o}s-R\'enyi graph are approximately the same size in the bulk of the spectrum. However, there are key challenges specific to our question. (i) Our function $f$ does not rely solely on $\bm{\varphi}$, but on $A$ as well. (ii) The polynomial approximation in \cite{huang2020size} is found using the closure of univariate polynomials in a specific Sobolev norm. Our function testing for nodal domains must be multivariate, as we need to control the sign of $\binom k2$ edges at once. (iii) We need to worry about the \emph{overlap} of different nodal domains. For example, merely checking the number of size $k$ nodal domains is not sufficient, as all of these could overlap on one vertex.

To solve (i), we have the added step of showing that $A$ and $\bm{\varphi}$ are close to independent in $f$. Because $\pphi$ is delocalized with high probability, perturbing a small number of entries does not significantly change $\pphi$ or $p_s$. We do this by Taylor expanding products of eigenvector entries. For (ii) we use results concerning the density of univariate polynomials in Sobolev norms of Rodr\'iguez \cite{rodriguez2003approximation}. These results do not generalize to all multivariate polynomials, but nevertheless, we can interpret our function $f$ as a composition of univariate functions, and show that approximation of each of these univariate functions is sufficient.  For (iii), rather than count the number of sets of size $k$ that are nodal domains, we count the number of sets of size $s$ that \emph{contain} a nodal domain, for $s\gg k$. There is some delicacy needed in that we require a set size $s$ that is large enough such that almost all sets contain a nodal domain, but small enough that when no sets of size $s$ that have a nodal domain are left, there are only few vertices left. This tightness is shown by Janson's Inequality \cite{janson2011random}.

Results in quantum ergodicity concern finite degree polynomials. Therefore we must consider constant $k$ rather than $k$ up to $\log n$, where we would expect these results to remain true. Moreover, the high probability statement is not strong enough to union bound over all indices at once. We suspect that the eigenvector entries are independent enough such that $N(\bm{\varphi})$ will emulate the chromatic number of a $G(n,p)$ graph and will be $\Theta(n/\log n)$.
{
\begin{conjecture}
Fix constants $0<\epsilon, p,q<1$. There are constants $c_1(\epsilon,p,q), c_2(\epsilon,p,q)$ such that with probability $1-o_N(1)$, for any index $i\in[\epsilon n,(1-\epsilon)n]$, the $i$th eigenvector $\pphi$ of a randomly sampled $G(n,p,q)$ satisfies $c_1 n/\log n\leq N(\pphi)\leq c_2 n/\log n$. 
\end{conjecture}
}
Note that we require our eigenvalue is in the ``bulk'' of the spectrum, in order to use quantum ergodicity. In Appendix \ref{sec:averagepath}, we compare this result to the path nodal domains of Mohammadian and Ge and Liu {and show that merely counting path nodal domains does not indicate a nontrivial relationship between the eigenvector and the structure of the graph.}

\begin{proposition}
    For fixed $0<p,q<1$, consider the adjacency matrix of the graph $G\sim G(n,p,q)$. With probability {$1-n^{-\omega(1)}$}, every eigenvector $\pphi$ satisfies $\kappa(G^<_{\bm{\varphi}})=1$.
\end{proposition}
This follows using the same proof as \cite{dekel2011eigenvectors}, using more recent eigenvector delocalization results, as is done in \cite{rudelson2017delocalization}.

We finish this discussion by giving a proof of the lower bound for $G(n,p,q)$. 
\begin{prop}\label{prop:avglower}
With high probability, for fixed $0<p,q<1$, any eigenvector $\bm{\varphi}$ of the adjacency matrix of a $G\sim G(n,p,q)$ has $\NN(\bm{\varphi})=\Omega(n/\log_{\frac{1}{1-(p\wedge q)}} n)$.
\end{prop}

\begin{proof}
Consider any set of vertices $S$ of size $k$. We will show that typically, there is no signing of eigenvectors that makes $S$ a nodal domain. With high probability, $\bm{\varphi}$ is nonzero \cite{nguyen2017random}. Take the signs of $\bm{\varphi}$ to be arbitrarily fixed. In order for a set of vertices $S$ to be a nodal domain, we must have $A_{ij}\pphi(i)\pphi(j)\geq 0$. Each edge satisfies this with probability at most $((1-p)\vee(1-q))$. Therefore the probability that $S$ forms a nodal domain is at most $((1-p)\vee(1-q))^{\binom k2}$. Union bounding over all possible signings, if we assume without loss of generality that $p\geq q$, the probability that $S$ forms a nodal domain is at most $2^k(1-q)^{\binom k2}= \exp(k\log 2+\binom k 2\log (1-q))$. Therefore the probability that there exists any such set is at most 
\[
\binom{n}k\exp\left(k\log 2+\binom k2\log (1-q)\right).
\]

As $\binom nk\leq n^k$, if, say, $k=3\log_{\frac{1}{1-q}} n$, then with probability $n^{-\Omega(\log n)}$ there are no nodal domains of size $k$.  As no nodal domain can have size $k$, there are at least $\Omega(n/\log_{\frac{1}{1-q}} n)$ domains. The desired result follows from a union bound over all $\pphi$. 
\end{proof}

\section{Classical Nodal Bounds}\label{sec:non-vanishing}

In this section, we provide tight upper and lower bounds on the nodal count of a non-vanishing eigenvector of a symmetric matrix in terms of the corresponding eigenvalue index and multiplicity, and a number of graph invariants. In addition, for vanishing eigenvectors, we prove the existence of orthonormal eigenbases with strong nodal count satisfying Courant-type nodal upper bounds. Finally, we illustrate that the set of eigenbases satisfying such conditions may be of measure zero.

\subsection{Non-Vanishing Nodal Count}\label{sub:generic}

Here, we prove tight bounds on the nodal count of a non-vanishing eigenvector. This result follows from \cite{mohammadian2016graphs} combined with an argument regarding the number of additional domains created by frustrated edges. However, we provide a direct proof by modifying a technique of Fiedler \cite{fiedler1975eigenvectors} and making use of known results regarding the inertia of signed Laplacian matrices. We have the following result.

\begin{proposition}\label{prop:generic_bound}
Let $M$ be a symmetric irreducible matrix and $\bm{\varphi}$ be a non-vanishing eigenvector corresponding to an eigenvalue of index $k$ and multiplicity $r$. Then
\begin{equation}\label{ineq:generic}
k +(r-1) - \mathsf{\nu}  \le \mathsf{N}(\bm{\varphi}) \le k + \mathsf{f},
\end{equation}
where $\nu$ and $\mathsf{f}$ are the cyclomatic number and frustration index of the signed graph of $M$.
\end{proposition}

\begin{proof}
Consider a symmetric, irreducible $n \times n$ matrix $M$ with eigenpair $(\lambda,\bm{\varphi})$, where $\lambda$ has index $k$ and multiplicity $r$, and $\bm{\varphi}$ is non-vanishing ($\bm{\varphi}(i) \ne 0$ for all $i$). The matrix $B=D_{\bm{\varphi}} (M-\lambda I) D_{\bm{\varphi}}$, where $D_{\bm{\varphi}}$ is a diagonal matrix with $\bm{\varphi}$ on the diagonal, is a signed Laplacian matrix, and, by Sylvester's law of inertia \cite[Thm. 8.1.17]{golub2013matrix}, has $n-k-r+1$ positive and $k-1$ negative eigenvalues.

We make use of the following result regarding the inertia of a signed Laplacian: let $\kappa_+$ and $\kappa_-$ be the number of connected components of the graph of a $n \times n$ signed Laplacian $L$ restricted to positive and negative entries, respectively; then 
$$ \kappa_+ - 1 \le \lambda^+ \le n - \kappa_- \qquad \text{and} \qquad \kappa_- - 1 \le \lambda^- \le n- \kappa_+,$$
where $\lambda^+$ and $\lambda^-$ are the number of positive and negative eigenvalues of $L$, respectively \cite[Thm. 2.10]{bronski2014spectral}.

Therefore, for $B=D_{\bm{\varphi}} (M-\lambda I) D_{\bm{\varphi}}$, we have\footnote{{Note that $\kappa_+$ and $\kappa_-$ of $B$ in the proof of Proposition \ref{prop:generic_bound} are exactly the quantities $G^>_{\bm{\varphi}}$ and $G^<_{\bm{\varphi}}$.}} 
$${ \kappa_+ - 1 \le n-k-r+1 \qquad \text{and} \qquad \kappa_- - 1 \le k-1 .}$$
Let $e_+$ and $e_-$ be the number of pairs of off-diagonal entries of $B$ that are positive and negative, respectively. Then $\kappa_+ \ge n- e_+$ and $\kappa_- \ge n - e_-$. {Using the above inequalities and recalling that $\nu = e-n+1$ and $e=e_+ + e_-$, we obtain our desired lower bound
\begin{align*}
\mathsf{N}_M(\bm{\varphi}) \ge \kappa_- &\ge n - e_- \\
&= n - e + e_+ \\ &\ge n - e + n - \kappa_+ \\ &\ge n - e + n - (n-k-r+2) \\ &=n - e + k + r -2 \\ &= k + (r-1) - \nu.
\end{align*}}
Suppose (w.l.o.g.) that $M$ has exactly $\mathsf{f}$ pairs of positive off-diagonal entries. Then, by choosing a nodal decomposition $[n] = \sqcup_{\ell = 1}^s V_\ell$ where entries $i$ and $j$ are in the same nodal subgraph only if $M_{ij}\le 0$ and $s$ is as small as possible, we have $\mathsf{N}_M(\bm{\varphi}) \le s \le \kappa_- + \mathsf{f} \le k+\mathsf{f}$, completing the proof.
\end{proof}

\begin{remark}\label{remark:invariants} {The proof of Proposition \ref{prop:generic_bound} actually produces a lower bound of $k+(r-1) -\nu + \nu_+ + \nu_-$, where $\nu_+$ and $\nu_-$ are the cyclomatic numbers of the graphs of $M$ restricted to entries where $M_{ij} \bm{\varphi}(i) \bm{\varphi}(j) >0$ and $M_{ij} \bm{\varphi}(i) \bm{\varphi}(j) <0$, respectively. However, in this work we attempt to focus on bounds in terms of graph invariants rather than quantities depending on the eigenvector itself. }
\end{remark}

Below we give a simple example illustrating the tightness of the bounds in Proposition \ref{prop:generic_bound} in general.

\begin{example}
Let $M$ be the negative adjacency matrix of the path on $n$ vertices, where $n+1$ is an odd prime, {with eigenpairs $\{(\lambda_k,\bm{\varphi}_k)\}_{k=1}^n$.} Consider $B = M + \epsilon C$, for some $\epsilon < 2/(n+1)^5$ and symmetric matrix $C$ with $|C_{ij}|\le 1$, $i,j \in [n]$. The minimal eigenvalue gap of $M$ is bounded below by $ \min_{\lambda,\lambda' \in \Lambda(M)} |\lambda - \lambda'| \ge 3\pi^2/4(n+1)^2$, and, because $n+1$ is an odd prime, every entry of each eigenvector of $M$ is bounded away from zero, namely
$$ |\bm{\varphi}_k(i)|  = \big| \sqrt{2/n} \sin[i k \pi/(n+1)] \big| \ge \pi /\sqrt{2}(n+1)^{3/2} \quad \text{for all } i,k \in [n].$$
Let $\{(\mu_k,\bm{\psi}_k)\}_{k=1}^n$ be the eigenpairs of $B$. By \cite[Cor. 8.1.6]{golub2013matrix}, $|\lambda_k - \mu_k| \le \epsilon n$, and so the spectrum of $B$ is simple and interlaces with that of $M$. In addition, by \cite[Thm. 8.1.12]{golub2013matrix}, $\| \bm{\varphi}_k-\bm{\psi}_k\| < \pi /\sqrt{2}(n+1)^{3/2}$, and so the eigenvectors of $B$ are also non-vanishing and have the same sign pattern as the eigenvectors of $M$. By Fielder's tree theorem \cite[Corollary 2.5]{fiedler1975eigenvectors}, $\mathsf{N_M}(\bm{\varphi}_k) = k$. To illustrate the tightness of the lower bound, we note that, for any $\nu\le k-2$, we may add $\nu$ edges to $M$ (through the matrix $C$), connecting $\nu$ pairs of non-adjacent nodal domains, and resulting in $\mathsf{N_B}(\bm{\psi}_k) = k - \nu$. For the upper bound, we note that when $k,\mathsf{f}\ll n$, each nodal domain is large and we may add non-crossing frustrated edges (with respect to the path ordering) within nodal domains, giving $\mathsf{N_B}(\bm{\psi}_k)=k+\mathsf{f}$.
\end{example}

\subsection{Orthonormal Eigenbases with Classical Strong Nodal Count} When an eigenvector has some vanishing entries, the above bounds no longer hold in general. By slightly modifying an argument of Mohammadian \cite{mohammadian2016graphs}, it is not hard to show that an upper bound of $\mathsf{N}^s(\bm{\varphi}) \le k + (r-1) + \mathsf{f}$ holds and is tight in general. However, by using a well-chosen orthonormal eigenbasis, we can obtain improved upper bounds, using a variation on a well-known technique {(see \cite[Sec. 3.2]{biyikoglu2007laplacian}, a variation on Courant's original technique \cite{courant1923allgemeiner,courant1924methoden}).}

\begin{proposition}\label{prop:strong_basis}
Let $M$ be an $n \times n$ symmetric matrix. There exists an orthonormal eigenbasis $\bm{\varphi}_1,...,\bm{\varphi}_n$ ordered by increasing energy, satisfying
\begin{equation}\label{ineq:strong_basis}
\mathsf{N}^s(\bm{\varphi}_k) \le k + \mathsf{f}, \qquad \qquad k = 1,...,n,
\end{equation}
where $\mathsf{f}$ is the frustration index of the signed graph of $M$.
\end{proposition}

\begin{proof}
It suffices to consider an $n \times n$ symmetric matrix $M$ with an eigenvalue $\lambda$ of index $k$ and multiplicity $r$, and produce an orthonormal basis $\bm{\varphi}_1,...,\bm{\varphi}_r$ of the corresponding eigenspace satisfying $\mathsf{N}^s(\bm{\varphi}_\ell) \le k + (\ell-1) + \mathsf{f}$ for $\ell = 1,...,r$. In addition, suppose, without loss of generality, that $M$ is the matrix in the equivalence class $\{ DMD \, | \, D \text{ involutory diagonal matrix} \}$ that minimizes the number of positive off-diagonal entries, i.e., the number of positive off-diagonal entries of $M$ equals $\mathsf{f}$. 

We proceed by induction. Suppose that we have orthonormal eigenvectors $\bm{\varphi}_1,...,\bm{\varphi}_\ell$, $\ell<r$, satisfying our desired nodal count {$\mathsf{N}^s(\bm{\varphi}_\tau) \le k +(\tau-1) + \mathsf{f}$ for $\tau = 1,...,\ell$} (if $\ell = 0$, we have no such vectors). Consider an arbitrary eigenvector $\bm{\varphi}$ orthogonal to $\bm{\varphi}_1,...,\bm{\varphi}_\ell$. Consider a minimal nodal decomposition $V_1,...,V_t$ of $[n] \backslash i_0(\bm{\varphi})$ for which $i$ and $j$ are in the same nodal domain only if $M_{ij}\le 0$ (i.e., a nodal decomposition satisfying this condition of minimal size). Suppose that $t>k+\ell + \mathsf{f}$, otherwise we are already done. Let
$$\bm{x}_p(i) = \begin{cases} \bm{\varphi}(i)& i \in V_p \\ \; 0 & \text{otherwise}\end{cases}\, , \quad p = 1,...,t$$
and consider the set of vectors $\bm{x}_{\bm{\alpha}} = \sum_{p=1}^t \alpha_p \, \bm{x}_{p}$ in their span, parameterized by $\alpha_1,...,\alpha_t$. By the orthogonality of $\{\bm{x}_1,...,\bm{x}_t\}$, this set is a subspace of $\mathbb{R}^{n}$ of dimension $t$. Let  $\bm{\alpha}(i) := \alpha_p$ for $i \in V_p$, $p = 1,...,t$. If $\bm{x}_{\bm{\alpha}}$ is a unit vector, then
\begin{equation}\label{eqn:partition}
\bm{x}_{\bm{\alpha}}^T M \, \bm{x}_{\bm{\alpha}} = \lambda - \sum_{(i,j) \in E} \big[ M_{ij} \bm{\varphi}(i) \bm{\varphi}(j) \big] (\bm{\alpha}(i) -\bm{\alpha}(j))^2.
\end{equation}
If vertices $i$ and $j$ are in the same nodal domain, then $\bm{\alpha}(i) = \bm{\alpha}(j)$. So, the only pairs $(i,j)\in E$ for which $\big[ M_{ij} \bm{\varphi}(i) \bm{\varphi}(j) \big] (\bm{\alpha}(i) -\bm{\alpha}(j))^2$ is strictly negative are those for which $i$ and $j$ are in different nodal domains, say $i \in V_p$ and $j \in V_q$, and there exists some $i^* \in V_p$ and $j^* \in V_q$ such that $M_{i^*, j^*} > 0$. There are $\mathsf{f}$ edges with $M_{i,j}>0$, and so the subspace of vectors $\bm{x_\alpha}$ for which $\bm{\alpha}(i) = \bm{\alpha}(j)$ for all $M_{ij}>0$ is of dimension at least $t - \mathsf{f} >k+\ell$. In addition, in this subspace, the Rayleigh quotient of all vectors $\bm{x_\alpha}$ is at most $\lambda$. By restricting our $\bm{x_\alpha}$ further to be orthogonal to $\bm{\varphi}_1,...,\bm{\varphi}_\ell$ and the eigenspaces of the $k-1$ eigenvalues strictly less than $\lambda$, we are left with a subspace of dimension of at least $t-\mathsf{f}-(k-1)-\ell>1$, which consists solely of eigenvectors of $\lambda$ orthogonal to $\bm{\varphi}_1,...,\bm{\varphi}_\ell$. This implies that there exists an $\bm{\widehat \alpha}$ in this subspace with $\alpha_1 = ... = \alpha_{t-\mathsf{f}-k-\ell} = 0$, and therefore $\bm{x_{\widehat\alpha}}$ has $\mathsf{N}^s(\bm{x_{\widehat\alpha}}) \le t - (t-\mathsf{f}-k-\ell) = k + \ell + \mathsf{f}$, completing the proof.
\end{proof}

Finally, by simply analyzing the Laplacian of a star graph, we note that the eigenbases satisfying the above proposition may be of measure zero. 

\begin{example}\label{ex:star}
    Let $M = \sum_{i=2}^n (\bm{e}_1 - \bm{e}_i)(\bm{e}_1 - \bm{e}_i)^T$, e.g., the graph Laplacian of a star. This matrix has eigenvalues $0$, $1$, and $n$, of multiplicity $1$, $n-2$, and $1$, respectively. The eigenvalue $\lambda = 1$ has eigenspace $E_1 = \{ \bm{x} \, | \, \sum_{i = 2}^n \bm{x}(i) = 0, \, \bm{x}(1) = 0 \}$. The $\bm{x} \in E_1$ satisfying $\mathsf{N}^s(\bm{x}) \le 2$ must have all but two entries equal to zero in an eigenspace of dimension $n-2$.
\end{example}
In the following section, we address this limitation by proving a more robust theorem regarding signings of eigenvectors.

\section{Orthonormal Eigenbases Satisfying Non-Vanishing Nodal Bounds}\label{sec:eigenbasis}

In this section, we prove Theorem \ref{thm:basis_bound}, breaking our analysis into three parts (as detailed in Section \ref{sec:intro}). First, we analyze the structure of repeated eigenvalues whose corresponding eigenvectors all vanish on some set of coordinates, Then, we place sign restrictions on an orthonormal basis of such an eigenspace so that, if such a basis exists, our desired nodal counts will be satisfied. Finally, we show that such non-vanishing orthonormal bases do exist and constitute a positive proportion of all orthonormal eigenbases. The proof of Theorem \ref{thm:basis_bound} is algorithmic in nature. For illustrative purpose, an example of this algorithm applied to a small matrix is given in Appendix \ref{sec:decompositionexample}.

\subsection{Part I: Structure of Eigenspaces with Vanishing Entries}\label{sub:basis_struct}

Let $M$ be a symmetric, irreducible $n\times n$ matrix with eigenvalue $\lambda$ of index $k$ and multiplicity $r$, and corresponding eigenspace $E_\lambda$. Recall that  $i_0(\bm{x}) := \{ j \in [n] \, | \,  \bm{x}(j) = 0 \}$, and let $$i_0(\lambda) = \{ j \in [n] \, | \, \forall \bm{\varphi}\in E_\lambda, \,\bm{\varphi}(j) = 0  \}.$$
We note that $i_0(\bm{\varphi}) = i_0(\lambda)$ for all but a set of positive co-dimension of $\bm{\varphi} \in E_\lambda$. Suppose $G[i_0(\lambda)]$ has connected components on $p$ vertex sets $X_1,...,X_p \subset i_0(\lambda)$ and $G\big[[n]\backslash i_0(\lambda)\big]$ has connected components on $q$ vertex sets $Y_1,...,Y_q \subset [n]\backslash i_0(\lambda)$. Let us write the matrix $M$ and an arbitrary eigenvector $\bm{\varphi} \in E_\lambda$ in block notation with respect to the partition $X_1,...,X_p,Y_1,...,Y_q$:
$$M = \begin{pmatrix} N & A \\ A^T & \hat M \end{pmatrix}, \quad \bm{\varphi} = \begin{pmatrix} 0 \\ \bm{\hat \bm{\varphi}} \end{pmatrix},$$
where
$$ N = \begin{pmatrix} N^{(1)} & 0 & \dots  & 0 \\   0 & N^{(2)}& \ddots  & \vdots \\  \vdots & \ddots & \ddots & 0 \\ 0 & \dots & 0 & N^{(p)} \end{pmatrix}, \quad \hat M = \begin{pmatrix} M^{(1)} & 0 & \dots  & 0 \\   0 & M^{(2)}& \ddots  & \vdots \\  \vdots & \ddots & \ddots & 0 \\ 0 & \dots & 0 & M^{(q)} \end{pmatrix},$$
are block diagonal matrices, $N^{(i)} \in \mathbb{R}^{|X_i| \times |X_i|}$ for $i = 1,...,p$, and $M^{(j)} \in  \mathbb{R}^{|Y_j| \times |Y_j|}$ for $j = 1,...,q$, and 
$$ A =\begin{pmatrix} A^{(1,1)} & \dots & A^{(1,q)} \\   \vdots & \ddots &  \vdots \\ A^{(p,1)}  & \dots & A^{(p,q)}  \end{pmatrix}, \quad \bm{\hat \bm{\varphi}} = \begin{pmatrix} \bm{\varphi}^{(1)} \\ \vdots \\ \bm{\varphi}^{(q)} \end{pmatrix} ,$$

$A^{(i,j)} \in \mathbb{R}^{|X_i| \times |Y_j|}$ for $i = 1,...,p$, $j = 1,...,q$, and $\bm{\varphi}^{(j)} \in \mathbb{R}^{|Y_j|}$ for $j = 1,...,q$. 
Next, we define $H = (X,Y,E_H)$ to be the bipartite graph, with bipartition $X = \{x_1,...,x_p\}$ and $Y = \{y_1,...,y_q\}$, representing the connectivity between the elements of $\{X_i\}_{i=1}^p$ and $\{Y_j\}_{j=1}^q$, i.e., $E_H = \big\{ (x_i,y_j)  \, | \, A^{(i,j)} \ne 0 \big\}$. {We note that $H$ is connected, as it is an aggregation of the connected graph $G$.} Let us define
\begin{align*}
u(i)&:= \text{smallest index } j \in [q] \text{ such that } x_i \sim_{H} y_j, \\
v(j)&:= \text{smallest index } i \in [p] \text{ such that } x_i \sim_{H} y_j.
\end{align*}
We order the elements of $X$ and $Y$ so that, for any $i_1>1$, $d_H(x_{i_1},x_{i_2}) = 2$ for some $i_2<i_1$, and $v(j_1) \le v(j_2)$ if $j_1 < j_2$. In addition, we suppose that our original ordering of vertex sets $X_1,...,X_p$ and $Y_1,...,Y_q$ corresponds to the aforementioned ordering of $x_1,...,x_p$ and $y_1,...,y_q$ in $H$. See Figure \ref{fig:graph} for an example of such a bipartite graph for a small matrix (a full example illustrating our procedure for this matrix is provided in Appendix \ref{sec:decompositionexample}.

\begin{figure}
{\centering
\includegraphics[width = 5.5 in]{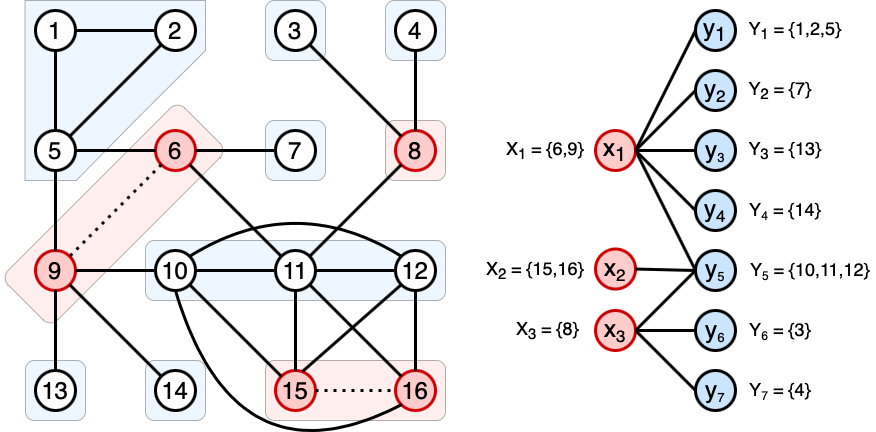}
\caption{The signed graph $G = ([16],E,\sigma)$ associated with the example matrix $M$ used in Appendix \ref{sec:decompositionexample}, and the corresponding bipartite graph for the eigenspace corresponding to $\lambda = 0$. The matrix has diagonal entries equal to $-1$ for $i = 1,2,5,10,11,12$, has off-diagonal entries equal to $+1$ for dashed edges and $-1$ for solid edges, and zeros otherwise. The vertices $6,8,9,15,16$ vanish on the eigenspace of $\lambda = 0$. See Appendix \ref{sec:decompositionexample} for the full analysis associated with this matrix.}\label{fig:graph}}
\end{figure}

From analysis of the eigenvalue-eigenvector equation, we note that $E_\lambda$ can be equivalently represented as the set of vectors $\bm{\varphi}^{(1)},...,\bm{\varphi}^{(q)}$ satisfying
$$\sum_{j=1}^q A^{(i,j)} \bm{\varphi}^{(j)} =\mathbf{0},\quad i = 1,...,p, \qquad \text{ and } \qquad M^{(j)} \bm{\varphi}^{(j)} = \lambda \, \bm{\varphi}^{(j)}, \quad j = 1,...,q.$$
Let $\bm{\psi}^{(j)}_1,...,\bm{\psi}^{(j)}_{r_j}$ be a non-vanishing orthonormal basis for the orthogonal projection of $E_\lambda$ to the indices of $Y_j$, $j = 1,...,q$. This projection (restricted to the indices of $Y_j$) is a subspace of the eigenspace of the matrix $M^{(j)}$ and eigenvalue $\lambda$. Let 
$\hat E_\lambda = \text{span} \{\bm{\psi}^{(j)}_{\sigma}\}_{\sigma=1,...,r_j}^{j = 1,...,q}$, $\hat k$ be the index of $\lambda$ with respect to $\hat M$, and $\hat r$ be the dimension of $\hat E_\lambda$. By eigenvalue interlacing, $\hat k + \hat r \le k + r$. Each eigenvector $\bm{\varphi} \in E_\lambda$ can be represented uniquely in the basis $\big\{\bm{\psi}_\sigma^{(j)}\big\}_{\sigma=1,...,r_j}^{j=1,...,q}$, say, 
$\bm{\varphi} = \sum_{j=1}^q \sum_{\sigma = 1}^{r_j} \alpha_\sigma^{j} \bm{\psi}^{(j)}_\sigma$. In fact, we can associate $E_\lambda$ with the subspace of $\{ \alpha_{\sigma}^{j}\}_{\sigma=1,...,r_j}^{j = 1,...,q} \cong \mathbb{R}^{\hat r}$ satisfying
$$\sum_{j=1}^q \sum_{\sigma = 1}^{r_j} \alpha_\sigma^{j} \big[ A^{(i,j)}\bm{\psi}^{(j)}_\sigma \big] = \mathbf{0}, \quad i = 1,...,p,$$
or, by Gaussian elimination, $\gamma:=\hat r - r$ homogeneous equations
$$h_\ell \big[ \{\alpha_{\sigma}^{j}\}_{\sigma = 1,...,r_j}^{j=1,...,q} \big] = \sum_{j=1}^q \sum_{\sigma = 1}^{r_j} c_{\ell,\sigma}^{j}\,\alpha_\sigma^{j} = 0, \quad \ell = 1,...,\gamma,$$
for some constants $c_{\ell,\sigma}^{j} \in \mathbb{R}$, $\ell = 1,...,\gamma$, $j = 1,...,q$, $\sigma = 1,...,r_j$.

Recall our eigenspace $E_\lambda$ corresponds to the subspace of $\{ \alpha_{\sigma}^{j}\}_{\sigma=1,...,r_j}^{j = 1,...,q} \cong \mathbb{R}^{\hat r}$ satisfying the $\gamma$ homogeneous equations $h_\ell\big[\{\alpha_{\sigma}^{j}\}_{\sigma=1,...,r_j}^{j = 1,...,q}  \big] = 0$. Consider the matrix associated with these $\gamma$ equations, where row $i$ corresponds to the linear form $h_i$, $i = 1,...,\gamma$, and the columns $\{1,..., \hat r\}$ correspond to the variables $\{\alpha_{\sigma}^{j}\}_{\sigma=1,...,r_j}^{j = 1,...,q}$ listed in reverse order: $\alpha_{r_q}^q,...,\alpha_{1}^q,\alpha_{r_{q-1}}^{q-1},...,\alpha_{1}^{q-1},...,\alpha_{r_1}^1,...,\alpha_{1}^1$ (e.g., column one corresponds to $\alpha_{r_q}^q$, column $\hat r$ to $\alpha_1^1$). Let us further suppose that the system of equations $h_\ell\big[\{\alpha_{\sigma}^{j}\}_{\sigma=1,...,r_j}^{j = 1,...,q}  \big] = 0$ is in reduced row echelon form with respect to the aforementioned ordering of equations and variables. In this case, the pivot for $h_\ell[ \cdot]$ is given by
$$\eta_\ell,\sigma_\ell = \text{argmax}_{j,\sigma} (j+ \sigma/r_j) \bm{1}\{c_{\ell,\sigma}^{(j)} \ne 0\}, \quad \ell= 1,...,\gamma,$$
and we denote the set of pivots by $\Sigma:=\{(\eta_{\ell},\sigma_{\ell}) \, | \, \ell = 1,...,\gamma\}$. Let us fix the values of the variables corresponding to the $\gamma$ pivots so that the equations $h_\ell[\{\alpha_\sigma^{j}\}] = 0$, $\ell =1,...,\gamma$, are satisfied:
$$\alpha_{\sigma_\ell}^{\eta_\ell} = -  \sum_{(j,\sigma) \not \in \Sigma} c_{\ell,\sigma}^{j} \, \alpha_{\sigma}^{j}, \quad \ell = 1,...,\gamma.$$

With the values of pivot variables fixed, our eigenspace $E_\lambda$ is parameterized by the coefficients $\big\{\{\alpha_\sigma^{j}\}_{\sigma = 1,...,r_j}^{j = 1,...,q} \, \backslash \, \{\alpha_{\sigma_\ell}^{\eta_\ell} \}_{\ell=1}^\gamma \big\}$, and in what follows we work directly with this formulation:
 $$E_\lambda = \text{span} \bigg\{ \bm{\psi}_{\sigma}^{(j)} - \sum_{\ell=1}^\gamma c_{\ell,\sigma}^{j} \bm{\psi}_{\sigma_\ell}^{(\eta_\ell)} \, \bigg| \, (j, \sigma) \ne \Sigma  \bigg\}.$$
Finally, we note that, for two eigenvectors $\bm{\varphi}_{1},\bm{\varphi}_2 \in E_\lambda$ with coefficients $\{{}_1\alpha_\sigma^{j}\}$ and $\{{}_2\alpha_\sigma^{j}\}$, respectively,
\begin{equation}\label{eqn:innerprod}
\langle \bm{\varphi}_1,\bm{\varphi}_2 \rangle = \sum_{(j,\sigma) \not \in \Sigma} {}_1\alpha_{\sigma}^{j} \bigg[{}_2\alpha_{\sigma}^{j} + \sum_{(\iota,\omega) \not \in \Sigma}  {}_2\alpha_{\omega}^{\iota} \sum_{\ell=1}^\gamma c_{\ell,\sigma}^{j} c_{\ell,\omega}^{\iota}  \bigg].
\end{equation}

\subsection{Part II: Restrictions That Produce Classical Nodal Bounds}\label{sub:basis_proof}

Now that we have sufficiently characterized the structure of an eigenspace with vanishing entries, we are now prepared to restrict the choices of an orthonormal basis $\bm{\varphi}_1,...,\bm{\varphi}_r$ of $E_\lambda$ and the choices of signings of vanishing vertices (given by $\bm{\varepsilon}_1,...,\bm{\varepsilon}_r$) so that $\NN(\bm{\varepsilon}_s) \le k + (s-1) + \mathsf{f}$ for $s = 1,...,r$, where $k$ is the index of $\lambda$ and $\mathsf{f}$ is the frustration index of the signed graph of $M$.

Here we make use of the notation introduced in Subsection \ref{sub:basis_struct}, and we assume $M$ minimizes the number of pairs of positive off-diagonal entries over the set $$\{ DMD \, | \, D \text{ involutory diagonal matrix }\}.$$
Therefore there are exactly $\mathsf{f}$ positive off-diagonal entries. Let us define:
\begin{align*}
\mathsf{\hat f} &= \text{number of pairs of positive off-diagonal entries of } \hat M, \\
    \mathsf{\tilde f} &= \text{number of pairs of positive off-diagonal entries $\{M_{a,b},M_{b,a}\}$} \\ &\qquad \text{where $a \in X_i$ for some $i$ and $b \in X_i \cup Y_{u(i)}$}.
\end{align*}
By applying Proposition \ref{prop:generic_bound} to each of the $q$ connected components of $G[[n]\backslash i_0(\lambda)]$ and using the inequality $\hat k + \hat r \le k + r$, any non-vanishing eigenvector of $\bm{\hat \bm{\varphi}} \in \hat E_\lambda$ satisfies (for an arbitrary integer $s\geq0$) \begin{equation}\label{ineq:subvec_bound}
    \mathsf{N}(\bm{\hat\bm{\varphi}}) \le \hat k + (q-1) + \mathsf{\hat f} \le k + (s-1) + \mathsf{f} + \big[q- \gamma - s - (\mathsf{f}-\mathsf{\hat f})\big].
\end{equation}
Furthermore, as noted in the proof of Proposition \ref{prop:generic_bound}, each eigenvector $\bm{\hat \bm{\varphi}}$ has a nodal decomposition satisfying \eqref{ineq:subvec_bound} where vertices $i$ and $j$ are in the same nodal subgraph only if $M_{ij}\le 0$. In what follows, we always assume that any nodal decomposition of $\bm{\hat \bm{\varphi}}$ under consideration also satisfies this property. Let us denote the coefficients corresponding to the eigenvector $\bm{\varphi}_s$ by ${}_s\alpha_{\sigma}^{j}$.  For each $(\bm{\varphi}_s,\bm{\varepsilon}_s)$, we aim to fix the signs $\bm{\varepsilon}_s(j)$ of vanishing entries $j \in i_0(\lambda)$ and restrict the values of the elements ${}_s\alpha_{\sigma}^{j}$ so that the original nodal count is decreased by at least $q- \gamma - s -(\mathsf{f}-\mathsf{\hat f})$ (using $q-\gamma-s+1$ elements) and $\bm{\varphi}_s$ is orthogonal to $\bm{\varphi}_t$ for all $t < s$ (using $s-1$ elements).

First, we describe our signing of $i_0(\lambda)$ for each vector. Consider an arbitrary non-vanishing vector $\bm{\hat \bm{\varphi}} \in \hat E_\lambda$. We restrict our signing so that $\bm{\varepsilon}$ is constant over each $X_i$, $i=1,...,p$. In particular, we set:
$$ \bm{\varepsilon}(a) = \sgn(\bm{\varphi}(b)) \text{ for all $a \in X_i$ and some fixed $b \in Y_{u(i)}$ in the neighborhood of $X_i$}.$$
Ignoring edges between $X$ and $Y$ that are not between $X_i$ and $Y_{u(i)}$ for some $i$ (i.e., each $X_i$ is connected to only one $Y_j$), the above signing of $i_0(\lambda)$ increases the nodal count by at most $\mathsf{\tilde f}$, already giving the bound
\begin{equation}\label{ineq:pessimistic}
\NN(\bm{\varepsilon}) \le \NN(\bm{\hat \bm{\varphi}}) +\mathsf{\tilde f}  \le  k + (s-1) + \mathsf{f} + \big[q- \gamma - s - (\mathsf{f}-\mathsf{\hat f}-\mathsf{\tilde f})\big].
\end{equation}
Next we will choose the signs of some of our coefficients ${}_s \alpha_\sigma^j$ so that the nodal count decreases further by using edges between $X$ and $Y$ that are not between $X_i$ and $Y_{u(i)}$ for some $i$. In  particular, we restrict our basis coefficients ${}_s\alpha_{\sigma}^{j}$ so that the nodal count of the $s^{th}$ eigenvector is at least $q- \gamma - s - (\mathsf{f} - \mathsf{\tilde f} - \mathsf{\hat f})$ less than the bound of Inequality \ref{ineq:pessimistic}. We need only consider $s < q-\gamma$, otherwise our desired bound already holds. Let us partition the variables ${}_s\alpha_{\sigma}^{j}$ into four sets, based on their function in the analysis that follows:
\begin{align*}
    \Pi_E^s &= \text{variables fully restricted so that $\bm{\varphi}_s \in E_\lambda$}, \\
    \Pi_S^s &= \text{variables restricted in sign so that $\NN(\bm{\varepsilon}_s) \le k + (s-1) + \mathsf{f}$}, \\
    \Pi_O^s &= \text{variables fully restricted so that $\langle \bm{\varphi}_s,\bm{\varphi}_t \rangle =0$ for all $t<s$}, \\
    \Pi_F^s &= \text{unrestricted variables}.
\end{align*}

Let $\hat Y = \{ y_j \, | \, j \ne \eta_\ell, \, \ell = 1,...,\gamma\}$. We note that $|\widehat Y| \ge q - \gamma$, and denote the indices of the first $q-\gamma$ elements of $Y$ in $\widehat Y$ by $j_1<...<j_{q-\gamma}$. 
For $s < q - \gamma$, we can set
\begin{align*}
    \Pi_E^s &= \big \{ {}_s\alpha_{\sigma}^j \, | \, (j,\sigma) \in \Sigma \big \}, \\
    \Pi_S^s &= \big \{ {}_s\alpha_{\sigma}^{j_m} \, | \, m = 2,..., q- \gamma -s + 1, \, \sigma = 1,...,r_{j_m} \big\}, \\
    \Pi_O^s &=\big \{ {}_s\alpha_{r_{j_m}}^{j_m} \, | \, m = q- \gamma -s + 2,...,q-\gamma \big\}, \\
    \Pi_F^s &= \{{}_s\alpha_\sigma^{j}\}_{\sigma = 1,...,r_j}^{j = 1,...,q} \, \backslash \, \big[ \Pi_E^s \sqcup \Pi_S^s \sqcup \Pi_O^s\big].
\end{align*}
By definition, when $s =1$, $\Pi_{O}^s = \emptyset$. We aim to show that the nodal count is at least $\big[q- \gamma - s - (\mathsf{f} -\mathsf{\tilde f} -\mathsf{\hat f})\big]$ less than Inequality \ref{ineq:pessimistic} for each $\bm{\varepsilon}_s$ by traversing the elements $ y_{j_1},...,y_{j_{q-\gamma}}$, and restricting the signs of the variables in $\Pi_S^s$ in some way.

For each pair $j_m,v(j_m)$, $m = 2,..., q - \gamma$, let $a_m,b_m \in [n]$ be a pair of vertices satisfying $a_m \in Y_{j_m}$, $b_m \in X_{v(j_m)}$, and $a_m \sim_G b_m$.
For every $\big\{{}_s \alpha_{\sigma}^{j_m} \big\}_{\sigma =1}^{r_{j_m}} \subset \Pi_{S}^s$, we require
$$\sgn \bigg( \sum_{\sigma =1}^{r_{j_m}} {}_s \alpha_{\sigma}^{j_m} \bm{\psi}_\sigma^{(j_m)}(a_m) \bigg) =  \bm{\varepsilon}_s(b_m) ,$$
resulting in $a_m$ and $b_m$ being in the same nodal domain if there is no frustrated edge between $Y_{j_m}$ and $X_{v(j_m)}$. We claim that these restrictions are consistent (e.g., there exist vectors simultaneously satisfying all conditions), and sufficient to produce our desired nodal bounds.

\begin{claim}\label{claim:bound}
Suppose an eigenvector $\bm{\varphi} = \sum_{(j,\sigma) \ne \Sigma} \alpha_{\sigma}^j \big[  \bm{\psi}_{\sigma}^{(j)} - \sum_{\ell=1}^\gamma c_{\ell,\sigma}^{j} \bm{\psi}_{\sigma_\ell}^{(\eta_\ell)} \big] \in E_\lambda$ and signing $\bm{\varepsilon} \in \{\pm 1\}^n$ satisfy
\begin{enumerate}
    \item {$\bm{\varepsilon}(i) = \sgn(\bm{\varphi}(i))$ for all $\bm{\varphi}(i) \ne 0$, $i \in [n]$,}
    \item $\bm{\varphi}(i) \ne 0$ for all $i \in [n] \backslash i_0(\lambda)$,
    \item $ \bm{\varepsilon}(a) = \sgn(\bm{\varphi}(b))$  for all $a \in X_i$ and some fixed $b \in Y_{u(i)}$ in the neighborhood of $X_i$,
    \item $\sgn \big( \sum_{\sigma =1}^{r_{j_m}} \alpha_{\sigma}^{j_m} \bm{\psi}_\sigma^{(j_m)}(a_m) \big) =  \bm{\varepsilon}(b_m) $ for $m = 2,...,q-\gamma-s+1$, for some $s>0$.
\end{enumerate}
Then $\NN\big(\bm{\varepsilon}\big) \le k + (s-1) + \mathsf{f}$. Furthermore, there exists some $(\pphi,\bm{\varepsilon})$ satisfying the above conditions for $s = 1$.
\end{claim}

\begin{proof}
We break our proof of the desired claim into two parts: first, we show the existence of $(\pphi,\bm{\varepsilon})$ satisfying the conditions of the claim for $s=1$, and then we show that satisfying Properties (1)-(4) for an arbitrary $s$ implies $\NN\big(\bm{\varepsilon}\big) \le k + (s-1) + \mathsf{f}$. 

We first aim to show that Properties (3) and (4) are consistent with each other, e.g., we can choose $\alpha_{\sigma}^j$, $(j,\sigma) \in \Sigma$, and $\bm{\varepsilon}$ so that both properties simultaneously hold. By the construction of our pivots, any $\alpha_{\sigma_\ell}^{\eta_\ell}$ is a linear function of variables of lower index, e.g., variables $\alpha_\sigma^j$ with $j\le \eta_\ell$, and, if $j = \eta_\ell$, then $\sigma<\sigma_\ell$. Therefore, $\bm{\varepsilon}|_{X_i}$ is a function of $\alpha_\sigma^j$ with $j \le u(i)$. Property (4) is a sign restriction on the variables $\{\alpha_{\sigma}^{j_m}\}_{\sigma = 1}^{r_{j_m}}$ corresponding to $Y_{j_m}$, depending on the quantity $\bm{\varepsilon}|_{X_{v(j_m)}}$, which depends only on $\alpha_\sigma^j$ with $j \le u(v(j_m))$. Therefore, it suffices to show that $j_m > u(v(j_m))$ for $m >1$. 

Recall that $d_H(x_{i_1},x_{i_2}) = 2$ for all $i_1>1$ and some $i_2<i_1$, and that $v(j_1) \le v(j_2)$ if $j_1 < j_2$. Because $m >1$, we may consider $j_{m-1}<j_m$, and note that $v(j_{m-1})\le v(j_m)$. If $v(j_{m-1})=v(j_m)$, then $u(v(j_{m})) \le j_{m-1} < j_m$, proving the desired result. If $v(j_{m-1})<v(j_m)$, then $v(j_m) \ne 1$, and so there exists some $i< v(j_m)$ with $d_H(x_i,x_{v(j_m)}) =2$. This implies that there is some $j$ with $y_j \sim_H x_i$ and $y_j \sim_H x_{v(j_m)}$, giving $v(j) \le i < v(j_m)$, implying $j < j_m$ and so $u(v(j_m)) \le j < j_m$. Therefore, we may indeed choose $\alpha_{\sigma}^j$, $(j,\sigma) \in \Sigma$, and $\bm{\varepsilon}$ so that both Properties (3) and (4) simultaneously hold.

What remains is to formally choose $\pphi$ so that Properties (2) and (4) hold. Each entry $\pphi(i)$, $i \not \in i_0(\lambda)$, is a non-trivial linear function in $\alpha$:
$$\bm{\varphi}(i) = \sum_{(j,\sigma) \ne \Sigma} \alpha_{\sigma}^j \bigg[  \bm{\psi}_{\sigma}^{(j)}(i) - \sum_{\ell=1}^\gamma c_{\ell,\sigma}^{j} \bm{\psi}_{\sigma_\ell}^{(\eta_\ell)}(i) \bigg]  ,$$
and has a lexicographically largest pair $(j,\sigma)$ for which the coefficient corresponding to $\alpha_\sigma^j$ is non-zero. If $\alpha_\sigma^j$ is not this variable for any $i \not \in i_0(\lambda)$, then we simply set $\alpha_\sigma^j = 0$. We set the values of the remaining $\alpha_{\sigma}^j$ iteratively, starting with the smallest values of $j$ (and, conditional on $j$, the smallest values of $\sigma$). Consider some $\alpha_{\sigma}^j$ not yet set, with all variables corresponding to smaller pairs $(j',\sigma')$ already set to some fixed value. If $\sigma = r_j$ and $j = j_m$ for some $1<m<q-\gamma$, then $\alpha_{\sigma}^j$ is restricted to the half-line defined by Property (4). Let us consider the set of linear functions $\pphi(i)$ only in $\alpha_{\sigma}^j$ and variables corresponding to smaller pairs $(j',\sigma')$. With all variables of lower index set to fixed values, each function is a non-trivial linear function of only $\alpha_{\sigma}^j$.
By avoiding the at most $n-|i_0(\lambda)|$ choices of $\alpha_{\sigma}^j$ for which any of these linear functions can be zero (also choosing $\alpha_{\sigma}^j$ to satisfy Property (4) if $\sigma = r_j$ and $j = j_m$ for some $1<m<q-\gamma$), and recursing, we have constructed a pair $(\pphi,\bm{\varepsilon})$ satisfying the conditions of the claim.

What remains is to show that Properties (1)-(4) (for $s$ arbitrary) implies that $\NN\big(\bm{\varepsilon}\big) \le k + (s-1) + \mathsf{f}$.  Recall, by Proposition \ref{prop:generic_bound} and eigenvalue interlacing, that
    \begin{equation}\label{ineq:g_prime}
    \NN\big( \bm{\varepsilon}|_{[n] \backslash i_0(\lambda)}  \big) = \NN\big( \bm{\varphi}|_{[n] \backslash i_0(\lambda)}  \big) \le k + (s-1) + \mathsf{f} + \big[q- \gamma - s - (\mathsf{f}-\mathsf{\hat f})\big] .
    \end{equation}
    Because $\bm{\varepsilon}$ is constant on each $X_i$ (by Property (3)), $\NN\big(\bm{\varepsilon}|_{X_i}\big)$ is at most one plus the number of positive pairs of off-diagonal entries of $M$ on the indices of $X_i$. In addition, because  $ \bm{\varepsilon}|_{X_i}$ has the same sign as $\bm{\varphi}(b)$  for some $b \in Y_{u(i)}$ in the neighborhood of $X_i$ (again, by Property (3)), either $b$ is in the same nodal domain as some $a \in X_i$, or there is a positive edge between $X_i$ and $Y_{u(i)}$. Let $G' = ([n],E')$ be the subgraph of $G$ with
    $$E' = E \, \backslash \, \big\{ (a,b) \in E \, | \, a \in X_i, b \in Y_j, j \ne u(i) \big\}.$$
    The subgraph of $H$ corresponding to $G'$, denoted by $H'$, is a forest consisting of trees, each with root in $Y$ and some number of leaves in $X$. Recall that $\mathsf{\tilde f}$ equals the number of pairs of positive off-diagonal entries either within some $X_i$ or between some $X_i$ and $Y_{u(i)}$. The nodal count $\NN_{G'}(\bm{\varepsilon})$ is then at most $\NN\big( \bm{\varepsilon}|_{[n] \backslash i_0(\lambda)}  \big) + \mathsf{\tilde f}$, as $\NN\big(\bm{\varepsilon}|_{X_i}\big)$ is at most one plus the number of positive pairs of off-diagonal entries within $X_i$, and either some vertex of $X_i$ is in the same nodal domain as a vertex of $Y_{u(i)}$, or there is a positive off-diagonal entry between $X_i$ and $Y_{u(i)}$. Combining this observation with Inequality \ref{ineq:g_prime} gives
    $$\NN_{G'}(\bm{\varepsilon}) \le \NN\big( \bm{\varepsilon}|_{[n] \backslash i_0(\lambda)}  \big) + \mathsf{\tilde f} \le  k + (s-1) + \mathsf{f} + \big[q- \gamma - s - (\mathsf{f}-\mathsf{\hat f} -\mathsf{\tilde f})\big].$$
    We have $j_m>u(v(j_m))$ for $m>1$, and so $j_m$ and $v(j_m)$ are in different trees of $H'$. Let $H'_t$ be the graph resulting from the addition of edges $\{(j_m,v(j_m)) \, | \, m = 2,...t \}$ to $H'$. More generally, $j_{t+1}$ and $v(j_{t+1})$ are in different connected components of $H'_t$ for any $t = 2,..., q-\gamma-s$. Now, let us consider the effect of the addition of each of the edges $(j_m,v(j_m))$, $m = 2,...,q-\gamma-s+1$, on our nodal bound (e.g., consider the nodal count of the sequence of graphs $H',H'_2,....,H'_{q-\gamma+s-1}$).  Property (4) implies that, for $m=2,...,q-\gamma-s+1$, either $a_m$ and $b_m$ are in the same nodal domain (implying that $H'_m$ has one less nodal domain than $H'_{m-1}$), or there is a frustrated edge of $M$ between $Y_{j_m}$ and $X_{v(j_m)}$, and so
    $$\NN_{G}(\bm{\varepsilon}) \le \NN_{G'}(\bm{\varepsilon}) -\big[q- \gamma - s - (\mathsf{f}-\mathsf{\hat f} -\mathsf{\tilde f})\big] \le k + (s-1) + \mathsf{f} , $$
    completing the proof.
\end{proof}

\subsection{Part III: Many Orthonormal Bases Satisfy the Conditions of Part II}

Claim \ref{claim:bound} shows that the conditions for our eigenvectors $\bm{\varphi}_1,...,\bm{\varphi}_{r}$ (and corresponding signings) are satisfiable and, if these conditions are satisfied, then the desired nodal bound is achieved. What's left is to show that the eigenvectors can simultaneously satisfy the conditions of Claim \ref{claim:bound} and be orthogonal to each other, and that there is a set of co-dimension zero of such orthonormal bases. We break the remainder of the argument into two parts: First, we show that we can build orthogonal eigenvectors $\bm{\varphi}_1,...,\bm{\varphi}_{q-\gamma-1}$ so that, for $\bm{\varphi}_s$, the conditions of Claim \ref{claim:bound} are satisfied for vertices $\sqcup_{i\le v(j_{q-\gamma-s+1})} X_i$ and $\sqcup_{j\le j_{q-\gamma-s+1}} Y_j$, but possibly not for the remaining vertices (Claim \ref{claim:orth}). Then, we show that such a set of orthogonal eigenvectors can be extended to an orthonormal basis, rotated so that the conditions of Claim \ref{claim:bound} hold for all vertices, and that these conditions are maintained under sufficiently small rotations of the orthonormal basis (Claim \ref{claim:main}).

\begin{claim}\label{claim:orth}
There exists eigenvectors $\bm{\varphi}_1,...,\bm{\varphi}_{q-\gamma-1}$, $\bm{\varphi}_s= \sum_{(j,\sigma) \ne \Sigma} {}_s\alpha_{\sigma}^j \big[  \bm{\psi}_{\sigma}^{(j)} - \sum_{\ell=1}^\gamma c_{\ell,\sigma}^{j} \bm{\psi}_{\sigma_\ell}^{(\eta_\ell)} \big] \in E_\lambda$, and signings $\bm{\varepsilon}_1,...,\bm{\varepsilon}_{q-\gamma-1} \in \{\pm 1\}^n$ such that $(\bm{\varphi}_1,\bm{\varepsilon}_1)$ satisfy the conditions of Claim \ref{claim:bound} for $s=1$, and $(\bm{\varphi}_s,\bm{\varepsilon}_s)$, $s = 2,...,q-\gamma-1$, satisfy
\begin{enumerate}
    \item  {$\bm{\varepsilon}_s(i) = \sgn(\bm{\varphi}_s(i))$ for all $\bm{\varphi}_s(i) \ne 0$, $i \in [n]$,}
    \item $\bm{\varphi}_s(i) \ne 0$ for all $i \in \sqcup_{j\le j_{q-\gamma-s+1}} Y_j$,
 \item $ \bm{\varepsilon}_s(a) = \sgn(\bm{\varphi}_s(b))$  for all $a \in X_i$ and some fixed $b \in Y_{u(i)}$ in the neighborhood of $X_i$, $i = 1,...,v(j_{q-\gamma-s+1})$,
\item $\sgn \big( \sum_{\sigma =1}^{r_{j_m}} \alpha_{\sigma}^{j_m} \bm{\psi}_\sigma^{(j_m)}(a_m) \big) =  \bm{\varepsilon}(b_m) $ for $m = 2,...,q-\gamma-s+1$,
\item $\langle \bm{\varphi}_s, \bm{\varphi}_t \rangle = 0$ for all $t<s$.
\end{enumerate}
\end{claim}

\begin{proof}
Claim \ref{claim:bound} guarantees the existence of our desired pair $(\bm{\varphi}_1,\bm{\varepsilon}_1)$. To prove this claim, we repeat a version of the proof of Claim \ref{claim:bound} in which the eigenvector $\bm{\varphi}_s$ needs only satisfy half-space and non-vanishing conditions for $Y_j$, $j \le j_{q-\gamma-s+1}$, has $\alpha_{\sigma}^j = 0$ for all elements of $\Pi_F^s$ with $j>j_{q-\gamma-s+1}$, and coefficients in $\Pi_O^s$ chosen so that $\bm{\varphi}_s$ is orthogonal to $\bm{\varphi}_t$, $t<s$.

We proceed by induction on the invertibility of the matrices associated with this orthogonalization procedure. In particular, given $\bm{\varphi}_1,...,\bm{\varphi}_{s-1}$, and some fixed choice of values for ${}_s\alpha_\sigma^j \not \in \Pi_O^s \cup \Pi_E^s$, orthogonality of $\bm{\varphi}_s$ to $\bm{\varphi}_t$, $t<s$, is equivalent (by Equation \ref{eqn:innerprod}) to the $s-1$ elements ${}_s\alpha_\sigma^j \in \Pi_O^s $ satisfying the $s-1$ linear equations
$$ \sum_{m = q-\gamma-s+2}^{q-\gamma}{}_s\alpha_{r_{j_m}}^{j_m} \bigg[{}_t\alpha_{r_{j_m}}^{j_m} + \sum_{(\iota,\omega) \not \in \Sigma}  {}_t\alpha_{\omega}^{\iota} \sum_{\ell=1}^\gamma c_{\ell,r_{j_m}}^{j_m} c_{\ell,\omega}^{\iota}  \bigg]  = -\sum_{{}_s\alpha_{\sigma}^{j} \not \in \Pi_O^s\cup \Pi_E^s} {}_s\alpha_{\sigma}^{j} \bigg[{}_t\alpha_{\sigma}^{j} + \sum_{(\iota,\omega) \not \in \Sigma}  {}_t\alpha_{\omega}^{\iota} \sum_{\ell=1}^\gamma c_{\ell,\sigma}^{j} c_{\ell,\omega}^{\iota}  \bigg]$$
for $ t= 1,...,s-1$. 

If the restrictions of $\bm{\varphi}_1,...,\bm{\varphi}_{s-1}$ to $\Pi_O^s$ are linearly independent, then the above system has a solution, as this restriction is a basis on $\Pi_O^s$, and so we may choose $\bm{\varphi}_{s}\big|_{\Pi_O^s}$ so that each $\langle \bm{\varphi}_t\big|_{\Pi_O^s} ,\bm{\varphi}_{s}\big|_{\Pi_O^s}\rangle$, $t<s$, is equal to any quantity we desire.

Therefore, it suffices to show that, at each step, the eigenvectors $\bm{\varphi}_1,...,\bm{\varphi}_{s-1}$ restricted to $\Pi_O^s$ are linearly independent. We begin with our base case of $s =2$. Our matrix is a scalar; we simply require that
$${}_1\alpha_{r_{j_{q-\gamma}}}^{j_{q-\gamma}} + \sum_{(\iota,\omega) \not \in \Sigma}  {}_1\alpha_{\omega}^{\iota} \sum_{\ell=1}^\gamma c_{\ell,r_{j_{q-\gamma}}}^{j_{q-\gamma}} c_{\ell,\omega}^{\iota}  \ne 0,$$
and note that the coefficient corresponding to ${}_1\alpha_{r_{j_{q-\gamma}}}^{j_{q-\gamma}}$ in the above linear function must be positive, as the eigenvector with $\alpha_{r_{j_{q-\gamma}}}^{j_{q-\gamma}} = 1$ and all other variables zero must have non-zero norm. By adding this single linear constraint to $\bm{\varphi}_1$, we have satisfied our base case of $s=2$. 

Now consider an arbitrary $s>2$. After selecting $\pphi_{s-1}$, if $\bm{\varphi}_1,...,\bm{\varphi}_{s-1}$ restricted to $\Pi_O^s$ are linearly independent, then we simply choose our coefficients for $\bm{\varphi}_s$ as in the proof of Claim \ref{claim:bound} for $Y_j$, $j \le j_{q-\gamma-s+1}$, and choose $\Pi_O^s$ to satisfy the above system. If the eigenvectors are linearly dependent, then, by induction $\bm{\varphi}_1\big|_{\Pi_O^s},...,\bm{\varphi}_{s-2}\big|_{\Pi_O^s}$ are linearly independent, and so $\bm{\varphi}_{s-1}\big|_{\Pi_O^s}$ is in the span of the other vectors. Let $\bm{x} \in E_\lambda$ be a vector orthogonal to $\bm{\varphi}_1,...,\bm{\varphi}_{s-2}$ such that $\bm{x} \big|_{\Pi_O^s} \not \in \text{span}\{\bm{\varphi}_1\big|_{\Pi_O^s},...,\bm{\varphi}_{s-2}\big|_{\Pi_O^s} \}$. Then, for a sufficiently small choice of $\delta$, $\sgn\big(\bm{\varphi}_{s-1}  + \delta \bm{x}\big)$ equals $\sgn\big(\bm{\varphi}_{s-1}\big)$ on $Y_j$, $j \le j_{q-\gamma-s+1}$. By simply replacing $\bm{\varphi}_{s-1}$ by $\bm{\varphi}_{s-1}  + \delta \bm{x}$, we now have our desired property, while maintaining all sign conditions.
\end{proof}

\begin{claim}\label{claim:main}
    Let $\mathcal{B}_\lambda$ be the set of orthonormal bases of $E_\lambda$. Then there exists a manifold $\Phi_\lambda \subset \mathcal{B}_\lambda$ of co-dimension zero, an ordering of basis elements $\{\bm{\varphi}_1,...,\bm{\varphi}_r\} \in \Phi_\lambda$, and signings $\bm{\varepsilon}_1,...,\bm{\varepsilon}_r \in \{\pm 1\}^n$ satisfying {$\bm{\varepsilon}_s(i) = \sgn(\bm{\varphi}_s(i))$ for all $\bm{\varphi}_s(i) \ne 0$, $i \in [n]$, $s = 1,...,r$,} such that
$$ \mathsf{N}(\bm{\varepsilon}_s) \le k+(s-1)+\mathsf{f}, \qquad \quad s =1,...,r.$$
    \end{claim}

\begin{proof}
By Claim \ref{claim:orth}, there are eigenvectors $\bm{\varphi}_1,...,\bm{\varphi}_{q-\gamma-1}$ that almost satisfy the conditions of Claim \ref{claim:bound}, but may vanish on 
$$[n] \backslash \big[ i_0(\lambda)  \bigsqcup \sqcup_{j\le j_{q-\gamma-s+1}} Y_j \big].$$
Consider an arbitrary extension and re-normalization of $\bm{\varphi}_1,...,\bm{\varphi}_{q-\gamma-1}$ to an orthonormal basis for $E_\lambda$: $\bm{\varphi}_1,...,\bm{\varphi}_{r}$. The vector $\bm{\varphi}_1$ is non-vanishing on $[n] \backslash i_0(\lambda)$, and using small rotations involving this vector, we may make every eigenvector in the basis non-vanishing on $[n] \backslash i_0(\lambda)$ without changing the sign of any non-zero entry. 

To do so, we make use of Givens rotations $G_{i,j}(\theta) \in \mathbb{R}^{n\times n}$, e.g., the orthogonal matrix with non-zero entries $\big[G_{i,j}(\theta)\big]_{kk} = 1$ for $k \ne i,j$, $\big[G_{i,j}(\theta)\big]_{kk} = \cos \theta$ for $k = i,j$, and $\big[G_{i,j}(\theta)\big]_{ij} = -\big[G_{i,j}(\theta)\big]_{ji} = \sin \theta$. Let $\rho$ be the magnitude of the smallest non-zero entry of $\bm{\varphi}_1,...,\bm{\varphi}_{r}$, and consider the following sequence of Givens rotations:
$$ \big[\bm{\varphi}'_1 \, ... \, \bm{\varphi}'_r  \big]= \bigg[\prod_{i=2}^{r} G_{1,i}(\rho/2^r)\bigg] \big[\bm{\varphi}_1 \, ... \, \bm{\varphi}_r \big] .$$
Any non-zero entry $\bm{\varphi}_s(j)$ has the same sign as $\bm{\varphi}'_s(j)$, as
\begin{align*}
\big| |\bm{\varphi}'_s(j)| - |\bm{\varphi}_s(j)| \big| &\le |\bm{\varphi}_s(j)| \big(1-\cos^{r-1}(\rho/2^r) \big) + (r-1)\sin(\rho/2^r ) \\
&\le \big(1- (1-(\rho/2^r)^2/2)^{r-1} \big) + (r-1)(\rho/2^r) \\
&\le 2^{r-1} (\rho /2^r)^2 + (r-1) (\rho/2^r) < \rho,
\end{align*}
and any zero entry $\bm{\varphi}_s(j)$, $j \not \in i_0(\lambda)$ results in a non-zero $\bm{\varphi}'_s(j)$, as each $s \ne 1$ has only one Givens rotation applied to it, and every entry $\bm{\varphi}_1(j)$, $j \not \in i_0(\lambda)$, is non-vanishing and remains so after each rotation.

Therefore, there exists a choice of $\{\bm{\varphi}_1,...,\bm{\varphi}_r\} \in \mathcal{B}_\lambda$ and a compatible set of signings $\bm{\varepsilon}_1,...,\bm{\varepsilon}_r \in \{\pm 1\}^n$ such that each eigenvector vanishes only on $i_0(\lambda)$, and $(\bm{\varphi}_s,\bm{\varepsilon}_s)$, $s < q - \gamma$, satisfy the conditions of Claim \ref{claim:orth}. Again, we recall that, for $s \ge q-\gamma$, by simply setting $ \bm{\varepsilon}(a) = \sgn(\bm{\varphi}(b))$  for all $a \in X_i$ and some fixed $b \in Y_{u(i)}$ in the neighborhood of $X_i$, $i = 1,...,p$, we automatically have 
$$\NN(\bm{\varepsilon}_s) \le k -\gamma  + q-1 + \mathsf{f} \le k +(s-1) + \mathsf{f}.$$
For $s<q-\gamma$, $(\bm{\varphi}_s,\bm{\varepsilon}_s)$, satisfying the conditions of Claim \ref{claim:orth} and not vanishing on $[n]\backslash i_0(\lambda)$ implies, by Claim \ref{claim:bound}, that $\NN(\bm{\varepsilon}_s) \le k +(s-1) + \mathsf{f}$ in this case as well.

To complete the proof, we simply expand the point $\big\{\bm{\varphi}_1,...,\bm{\varphi}_r\big\} \in \mathcal{B}_\lambda$ using rotations of arbitrarily small angle. In particular, we set
$$\Phi_\lambda = \bigg\{ \bigg[\prod_{i=1}^{\binom r2} G_{p_i,q_i}(\theta_i)\bigg] \big[\bm{\varphi}_1 \, ... \, \bm{\varphi}_r \big] \, \bigg| \, p_i,q_i \in[r], \, \theta_i \in [ -\upsilon,\upsilon], \, i = 1,...,r \bigg\},$$
where 
$$\upsilon < 2^{-{r \choose 2}} \min_{\substack{s \in 1,...,r \\ j \in [n] \backslash i_0(\lambda)}} |\bm{\varphi}_s(j)|.$$ 
Considering an arbitrary entry $\bm{\varphi}_s(j)>0$, we note that the composition of $r \choose 2$ Givens rotations with angles in $[-\upsilon,\upsilon]$ can change this entry by at most
\begin{align*}
    \big|\bm{\varphi}_s(j) - \big[ \cos^{r \choose 2}(\upsilon) \bm{\varphi}_s(j)  -  \sum_{i=1}^{r \choose 2}\sin(\upsilon)\cos^{i-1}(\upsilon) \big] \big| &\le |\bm{\varphi}_s(j)| (1-\cos^{r \choose 2}(\upsilon)) + \textstyle{r \choose 2 }\sin(\upsilon) \\
    &\le  |\bm{\varphi}_s(j)| (1-(1-\upsilon^2/2)^{r \choose 2}) + \textstyle{r \choose 2 }\upsilon \\
    &\le 2^{r \choose 2} \upsilon^2 + \textstyle{r \choose 2} \upsilon < \bm{\varphi}_s(j),
\end{align*} 
and so every basis in $\Phi_\lambda$ has the same sign pattern as $\{\bm{\varphi}_1,...,\bm{\varphi}_r\}$, and, therefore, the same $\bm{\varepsilon}_1,...,\bm{\varepsilon}_r$. This completes the proof. 
\end{proof}

\section{Nodal Count of Signed Erd\H os-R\'enyi Graphs}\label{sec:average}

 In this section, we prove Theorem \ref{thm:mainavg}. We consider an Erd\H{o}s-R\'enyi random signed graph $G(n,p,q)$, with $0<p,q<1$ fixed constants. We give $\bm{\varphi}$ the index $i$. We prove the following.

 {
 \begin{theorem}[Specific version of Theorem \ref{thm:mainavg}]
For any $0<\epsilon,p,q<1$, there is a constant $\gamma>0$ such that for every $\alpha>0$, there is some $N$ such that for $n>N$, index $i\in [\epsilon n,(1-\epsilon)n]$ and the $i$th eigenvector $\pphi$ of the adjacency matrix of $G\sim G(n,p,q)$, the probability that $\NN(\pphi)<\alpha n$ is at least $1-n^{-\gamma}$. 
 \end{theorem}
}

 Therefore, by taking an infinite decreasing sequence of $\alpha$, the number of nodal domains is $o(n)$ with probability $1- O(n^{-\gamma})$.
 
 In order to be consistent with the literature and reason about the spectrum, we work with both the adjacency matrix $A$ and a normalized version $\tilde A:=\frac{1}{\sqrt{(p+q)n}}A$,  so that $\E(\tilde A_{u,v}^2=1/n)$. 
We proceed with the steps as listed in the introduction.

\subsection{Part I: Function Definition}
In order to count nodal domains, we choose some $s$ such that almost all sets of $s$ vertices contain a nodal domain of size $k$ for $k>0$. We set 
\begin{equation}\label{eq:defs}
s:=(p\wedge q)^{-k}
\end{equation} and will show this is sufficient. To quantify whether a set $|S|=s$ contains a nodal domain of size $k$, we consider the function $f_s:\R^{\binom {s+1}2}\rightarrow\R$, defined as
\[
f_s(A_S,\sqrt{n}\bm{\varphi}(S)):= 1_{> 0}\left(\sum_{B\in \binom{S}k}\prod_{(u,v)\in \binom{B}{2}} 1_{>0}(n\bm{\varphi}(u)\bm{\varphi}(v)A_{uv})\right).
\]

Therefore $f_s$ asks whether $S$ contains a nodal domain that is a clique. We call these \emph{clique domains}. This is more specific than a general nodal domain, but is analytically easier to deal with. Moreover, we expect these form a constant portion of all nodal domains, so our asymptotic result does not change.  Similar analysis would work with a slightly more complicated function that counts nodal domains of any type, namely
\[
 1_{{>}0}\left[\sum_{B\in \binom {S} k}\left(\sum_{T\in {\mathfrak T_B}}\prod_{uv\in T} A_{uv}^2\right)\prod_{(u,v)\in \binom{B}{2}} 1_{\geq 0}(n\bm{\varphi}(u)\bm{\varphi}(v)A_{uv})\right]
\]
where $\mathfrak T_B$ is the set of spanning trees of $S$.

Note that we have renormalized $f_s$ so that, typically, its eigenvector inputs are $\Theta(1)$.

\subsection{Part II: Polynomial Approximation}
We approximate $f_s$ with a finite degree polynomial. We localize the distribution of $A$ to $A_S$ by defining $\mathbf{M}\subset \text{Mat}_{sym}(s)$ as the set of feasible assignments of $A_S$. Specifically,  $\mathbf{M}$ is the set of matrices $M$ such that for any pair of indices $u,v$
\[M_{uv}\in\left\{
\begin{array}{cc}
     \{0\}&u=v  \\
     \{0,\pm 1\}& u\neq v.
\end{array}\right.
\]
 We equip $\mathbf{M}$ with the distribution of $A_S$. For $M\in \mathbf M$ we will write $\tilde M:=\frac{1}{\sqrt{(p+q)n}}M$. From now on, we denote by $\mathbf{g}$ a length $s$ vector of i.i.d. standard normal Gaussians of length $s$. Abusing notation, we will also denote the $s$ dimensional probability measure of $\mathbf{g}$ by $\mathbf{g}$.

\begin{lemma}\label{lem:defp}
    For any $0<\delta<1, C>1$,  there is a finite degree polynomial $p_s:\R^{\binom{s+1}{2}}\rightarrow \R$, such that the following are true. We consider $y\in \R^{s}$ to stand in for the contribution of $\pphi$. Then
    \begin{enumerate}
    \item
    For any $M\in\mathbf{M}$, $p_s$ satisfies the following bounds.

        \begin{eqnarray}\label{eq:pinf}
        |p_s(M,y)-f_s(M,y)|\leq 
             \delta & y: \forall u,v\in \binom{[s]}{2}, y_{u}y_{v}\in[-C,{-\delta}]  \cup [\delta,C] \\
             \label{eq:pinf2}
            |p_s(M,y)|\leq 1+\delta & y: \forall u,v\in \binom{[s]}{2}, y_{u}y_{v}\in[-C,C] 
        \end{eqnarray}
        
    \item
    For any $M\in\mathbf{M}$, $p_s(M,y)$ is even in $y$ and
    \begin{equation}
    \label{eq:ell2bound}
    \E\left[p_s({M},\mathbf{g})^2\right]\leq 2.
    \end{equation}
    \end{enumerate}
\end{lemma}

Before we prove this lemma, we give an idea of the method. We will work in a weighted Sobolev normed space. Our weight function $\mu$ is a probability measure on $\R$ defined as 
\begin{equation}\label{eq:defmu}
d\mu(x):=\frac12e^{-|x|}.
\end{equation}

We then define the Sobolev norm as for any function $f:\R\rightarrow \R$,
\[
\|f\|_{W^{k,p}}=\left(\int_{-\infty}^\infty \sum_{i=0}^k |f(x)^{(i)}|^p d\mu(x)\right)^{1/p}.
\]
The key fact is that under this norm, $f_s$ can be approximated by polynomials. This, by a result of Rodr\'iguez, is implied by the fact that polynomials are dense in $L_p(\R,\mu)$ (see \cite{koosis1998logarithmic} Page 170).

\begin{lemma}\label{lem:sobolevdense}[\cite{rodriguez2003approximation} Proposition 4.2]
   For $\mu$ defined in \eqref{eq:defmu}, $k\in \N$ and $1\leq p<\infty$,
 polynomials are dense in $W^{k,p}$ among all functions that have bounded $W^{k,p}$ norm. 
\end{lemma}
{
  Bounding our Sobolev norm also bounds the infinity norm on intervals, as we can embed $W^{1,p}$ into the space of functions bounded on $[-C,C]$. The following is implied by (\cite{burenkov1998sobolev} 4.1 Lemma 1).
\begin{lemma}\label{lem:wpinfbound}
For $\mu$ defined in \eqref{eq:defmu} and $\|F\|_{W^{1,p}}\leq \epsilon$, there is some constant $c(C)$ such that
\[
\|F\vone_{[-C,C]}\|_{\infty}\leq c\cdot  \epsilon.
\]
 \end{lemma}
 }

\begin{proof}[Proof of Lemma \ref{lem:defp}]
We first approximate $f_s(M,y)$ with a differentiable version. Therefore define the function $\eta_\delta:\R\rightarrow \R$
\[
\eta_\delta(x)=\left\{\begin{array}{cc}
     0& x\leq 0 \\
     3x^2/\delta^2-2x^3/\delta^3 & x\in [0,\delta]\\
     1& x\geq \delta
\end{array}\right. .
\]
From now on, we write $r:=\binom{s}{k}$. Our differentiable approximation of $f_s(M,y)$ is 
\[
\eta_{1/2}\left(\sum_{\alpha\in[r]}\prod_{uv\in \binom\alpha2} \eta_{\delta}(M_{uv} y_uy_v)\right).
\]
The outer function receives parameter 1/2, considering we just want it to distinguish 0 from 1 or more. The inner function needs parameter $\delta$, as it takes the Gaussian input. 

We approximate $\eta_{1/2}$ and $\eta_\delta$ separately with polynomials $P$ and $Q$. We take our outer approximation first. Consider $C_2:=2\binom{s}k$. We approximate $\eta_{1/2}$ on the compact interval $[-C_2,C_2]$ with a polynomial $P$ so that $\|(P-\eta_{1/2})\vone_{[-C_2,C_2]}\|_\infty\leq \epsilon$ for some $\epsilon$ to be determined.

We denote by $d$ the degree of $P$. We choose $Q$ by having it approximate $\eta_\delta$ in $W^{1,{p}}$ for $p:=\binom k2 4d$. Using Lemma \ref{lem:sobolevdense} and Lemma \ref{lem:wpinfbound}, we can approximate $\eta_\delta$ to such accuracy that 
\begin{equation}\label{eq:Lp}
\|(Q-\eta_\delta)\vone_{[-2C,2C]}\|_\infty\leq \epsilon,~\|Q\|_{L_p}\leq 1.
\end{equation}
for some $\epsilon$ to be determined. This second inequality is possible as $\|\eta_\delta\|_{L_p}< 2^{-\frac1p}$. 

We can now show the properties of $P(\sum_{i\in [r]} Q(x_i))$. Denote by $\vone_C$ the event that $g_ug_v\in [-C,C]$ for each pair of vertices $u,v\in S$. For the infinity norm, we consider
\begin{equation}\label{eq:bigterm}
\left\|\left[\eta_{1/2}\left(\sum_{\alpha\in[r]}\prod_{uv\in \binom\alpha2} \eta_\delta(M_{uv}g_ug_v)\right)-P\left(\sum_{\alpha\in [r]}\prod_{uv\in \binom\alpha 2}Q(M_{uv}g_ug_v)\right)\right]\vone_{C}\right\|_{\infty}.
\end{equation}
By \eqref{eq:Lp}, and the fact that $\|\eta_\delta\|_\infty \leq 1$, we have that under $\vone_C$, 
$|\left(\sum_{\alpha\in [r]}\prod_{uv\in \binom \alpha 2}\eta_\delta\right)-\left(\sum_{\alpha\in [r]}\prod_{uv\in \binom \alpha 2}Q\right)|\leq r\big((1+\epsilon)^{\binom k2}-1\big)$. Using the approximation $(1+\epsilon)^{\binom{k}2}-1\leq k^2\epsilon$ for small $\epsilon$,
\begin{eqnarray*}
\eqref{eq:bigterm}&\leq &\max_{\substack{x\in [0,r]\\ |x-y|\leq \epsilon k^2 r}} |\eta_{1/2}(x)-P(y)|\\
&\leq &\max_{\substack{x\in [0,r]\\ |x-y|\leq \epsilon k^2r}} |\eta_{1/2}(y)-P(y)|+|\eta_{1/2}(x)-\eta_{1/2}(y)|\\
&\leq& \epsilon+3\epsilon k^2 r.
\end{eqnarray*}
For sufficiently small $\epsilon$, this satisfies \eqref{eq:pinf} and \eqref{eq:pinf2}. Now we will show \eqref{eq:ell2bound}. We start with
\begin{multline*}
\E\left[P\left(\sum_{\alpha\in [r]}\prod_{uv\in\binom{\alpha} 2}Q(M_{uv}g_ug_v)\right)^2\right]\leq \E\left[P\left(\sum_{\alpha\in [r]}\prod_{uv\in\binom{\alpha} 2}Q(M_{uv}g_ug_v)\right)^2 1_C\right]\\
+\E\left[P\left(\sum_{\alpha\in [r]}\prod_{uv\in\binom{\alpha} 2}Q(M_{uv}g_ug_v)\right)^2\overline 1_C\right].
\end{multline*}
For the first term on the right, we have by the infinity norm bound, 
\begin{eqnarray*}
\int_{} P\left(\sum_{\alpha\in [r]}\prod_{uv\in\binom{\alpha} 2}Q(M_{uv}g_ug_v)\right)^2 \vone_Cd\mathbf{g}&\leq&\epsilon+3\epsilon k^2 r
\end{eqnarray*}
For the other term, we use
\[
\E\left[P\left(\sum_{\alpha\in [r]}\prod_{uv\in\binom{\alpha} 2}Q(M_{uv}g_ug_v)\right)^2\overline 1_C\right]\leq \E\left[P\left(\sum_{\alpha\in [r]}\prod_{uv\in\binom{\alpha} 2}Q(M_{uv}g_ug_v)\right)^4\right]^{1/2}\pr(\overline 1_C)^{1/2}
\]
We can write
\[
P(x)^4=\sum_{m=0}^{4d} c_m x^m.
\]
This gives
\begin{eqnarray*}
\int_{\R^s} P\left(\sum_{\alpha\in [r]}\prod_{uv\in \binom\alpha 2}Q(M_{uv}g_u g_v)\right)^4d\mathbf{g}&=& \sum_{m=0}^{4d}\int_{\R^s} c_m \left(\sum_{\alpha\in [r]}\prod_{uv\in \binom\alpha 2} Q(M_{uv}g_u g_v)\right)^m d\mathbf{g}\\
&\leq& \sum_{m=0}^{4d}|c_m|r^{m-1} \sum_{\alpha\in[r]}\int_{\R^s} \prod_{uv\in\binom{\alpha} 2}|Q(M_{uv}g_u g_v)|^{m} d\mathbf{g}\\
&\leq& \binom{k}{2}^{-1}\sum_{m=0}^{4d}|c_m|r^{m-1} \sum_{\alpha\in[r]}\sum_{uv\in \binom \alpha 2}\int_{\R^s} |Q(g_ug_v)|^{m\binom k2}d\mathbf{g}.
\end{eqnarray*}
where at the last line we use the AM GM inequality and the symmetry of the measure.

Recall that the product of two i.i.d. standard normal Gaussians has the tail bound $\pr(|g_ug_v|\geq t)\leq 2e^{-t}$ (see \cite[Lemma 2.7.7]{vershynin2018high}). Therefore,
\begin{eqnarray*}
\int_{\R} |Q(g_ug_v)|^{m\binom k2}d\mathbf{g}&\leq&4\int_{-\infty}^\infty |Q(x)|^{m\binom k2} d\mu\\
&\leq &4\int_{-\infty}^\infty \max\{1,|Q(x)|^{p}\} d\mu\leq 4(1+1)
\end{eqnarray*}
by \eqref{eq:Lp}. As $\pr(\overline \vone_C)\leq 2re^{-|C|}$, for sufficiently large $C$,
\[
\left\|P\left(\sum_{\alpha\in [r]}\prod_{uv\in\binom{\alpha} 2}Q(M_{uv}g_ug_v)\right)\right\|_{L_{2}(\mathbf g)}\leq 1+\epsilon+\pr(\overline 1_C)^{1/2}\left(8\sum_m |c_m|r^m\right)\leq 2.
\]
\end{proof}

\subsection{Part III: Independence}

We use the following structural laws concerning the spectrum and eigenvectors. These structural results are a combination of \cite[Equation 4.11]{huang2015bulk}, \cite[Proposition 4.3]{bourgade2017eigenvector}, and \cite[Corollary 3.2]{erdHos2010wegner}.

\begin{lemma}\label{lem:delocalization}
{For $i$ as defined in Theorem \ref{thm:mainavg}, we consider eigenvalue $\lambda_i$ of $\tilde A$ with eigenvector $\pphi$. }
\begin{enumerate}
    \item For any fixed unit vector $w\perp \overrightarrow{1_{[n]}}$ and any $c>0$ there is a constant $\gamma>0$ such that with probability $1-O(n^{-\gamma})$,
    \begin{equation}\label{eq:lambdabound}
    \sum_{j\neq i}\frac{1}{|\lambda_j-\lambda_i|}\leq n^{1+c}
    \end{equation}
    and for every eigenvector $\pphi_j$ of $\tilde A$,
     \begin{equation}\label{eq:eigbound}
     |\langle w,\pphi_j\rangle|\leq n^{-1/2+c}.
     \end{equation}

  \item

   With probability $1-n^{-\omega(1)}$
   \begin{equation}\label{eq:infnormbound}
    \|\pphi\|_{\infty}\leq \log^{4}n/\sqrt{n}.
\end{equation}

\end{enumerate}
    
\end{lemma}

 We define $\omegam$ to be the high probability event that \eqref{eq:lambdabound},\eqref{eq:eigbound},\eqref{eq:infnormbound} are true {for index $i$}, $c<1/20$ and a finite set of vectors $w$ to be determined throughout the course of the proof.  

We wish to use these delocalization results to control the change in $\bm{\varphi}(S)$ while changing $A_S$. Therefore, consider the normalized block adjacency matrix of the $G(n,p,q)$ graph
\[
\tilde A=
\left[
\begin{array}{cc}
\tilde{A}_{\overline S} &\tilde{A}_{\overline S,S}\\
\tilde{A}_{S,\overline S} & \tilde{A}_S
\end{array}\right].
\]
Fixing the rest of $\tilde A$, we can replace $\tilde{A}_S$ with $\tilde M$ to create a new adjacency matrix. Namely, given $\tilde A$, we define a function $\psi^{\tilde A}(\tilde M):\text{Mat}_{sym}(s)\rightarrow\R^{s}$, where $\psi^{\tilde A}(\tilde M)=\pphi^{\tilde M}(S)$, for $\pphi^{\tilde M}$ the $i$th eigenvector of
\[
\left[\begin{array}{cc}
\tilde{A}_{\overline S }&\tilde{A}_{\overline S,S}\\
\tilde{A}_{S,\overline S} & \tilde M
\end{array}\right].
\]
Note that $\psi^{\tilde A}(\tilde{A}_S)=\pphi(S)$. We can now decouple the dependence of $A$ and $\bm{\varphi}$. {Generally speaking, because the eigenvector is delocalized, we can bound the change in the eigenvector from a perturbation to a small submatrix.}

\begin{lemma}\label{lem:decor}
Assume $\omegam$. Then for any finite degree even polynomial $F:\R^s\rightarrow \R$ and $M\in \mathbf{M}$,

\[
|F(\sqrt{n}\psi^{\tilde A}(\tilde M))-F(\sqrt{n}\pphi(S))|= O(n^{-1+3c}).
\]
\end{lemma}

\begin{proof}

We track the change in $\bm{\varphi}(u)\bm{\varphi}(v)$ for $u,v\in S$. Therefore we consider the function $\psi^{\tilde A}_{uv}(x):\R^{\binom{s}2}\rightarrow \R$ defined as
\[
\psi^{\tilde A}_{uv}(M):=[\psi^{\tilde A}(M)](u)[\psi^{\tilde A}(M)](v).
\]
For $M,M'\in \mathbf{M}$, $\|\tilde M'-\tilde M\|_{1\rightarrow \infty}\leq \frac{2}{\sqrt{(p+q)n}}$. 
Therefore taking the Taylor expansion at $x=A_S$,
\begin{multline}\label{eq:entexpansion}
|\psi^{\tilde A}_{uv}(M)-\pphi(u)\pphi(v)|\leq \sum_{k=1}^\infty \bigg(\frac{2}{\sqrt{(p+q)n}}\bigg)^{k}\sum_{x_1,\ldots x_k\in[\binom s2]}\frac{1}{k_1!k_2!\cdots k_{\binom{s}{2}}!}\left| \left[\frac{\partial^k}{\partial x_1,\ldots x_k}\psi^{\tilde A}_{uv}\right]( A_S)\right|.
\end{multline}
where $k_m$ is the number of times the $m$th edge is chosen.

In order to calculate the partial derivative, we proceed as per \cite[Section 4]{bourgade2017eigenvector}. Define $V_{x}$ to be the matrix with $1$'s in the off-diagonal coordinates corresponding to $x$, but $0$'s elsewhere. Moreover, we denote the Green's function by $G(z):=(A-z)^{-1}$. For a vector $w$, taking a contour integral around only the $i$th eigenvalue gives, by the Cauchy residue formula,
\[
\langle w,\varphi\rangle^2=\frac{1}{2\pi i}\oint \langle w,G(z)w\rangle dz.
\]
Therefore, by the Cauchy residue formula again,
\begin{eqnarray*}\frac{\partial^k}{\partial x_1,\ldots x_k}\langle w,\pphi\rangle ^2&=&\frac{(-1)^{k}k!}{2\pi i}\oint \langle w,\prod_{\ell\in [k]}(G(z)V_{x_\ell})G(z)w\rangle dz\\
&=&k(-1)^{k}k!\sum_{j\in ([n]\backslash \{i\})^k}\frac{\langle w,\bm{\varphi}_{j_1}\rangle\langle w,\bm{\varphi}\rangle(\pphi_{j_k}^TV_{x_k}\pphi)\prod_{\ell\in [k-1]} (\bm{\varphi}_{j_\ell}^TV_{x_\ell}\bm{\varphi}_{j_{\ell+1}})}{\prod_{\ell\in [k]}(\lambda_{j_\ell}-\lambda_i)}.
\end{eqnarray*}
Therefore, assuming $\omegam$, we have that for unit $w\perp\overrightarrow{1_{[n]}}$ and ignoring a coefficient only depending on $k$, 
\begin{eqnarray*}
|\partial_{ab}^{(k)}\langle w,\bm{\varphi}\rangle ^2|&\leq &k(k!)\left|\sum_{j\in ([n]\backslash \{i\})^k}\frac{\langle w,\bm{\varphi}_{j_1}\rangle\langle w,\bm{\varphi}\rangle(\pphi_{j_k}^TV\pphi)\prod_{\ell\in [k-1]} (\bm{\varphi}_{j_\ell}^TV\bm{\varphi}_{j_{\ell+1}})}{\prod_{\ell\in [k]}(\lambda_{j_\ell}-\lambda_i)}.\right|\\
&\leq &k(k!)2^k(\log n)^{8k}n^{-k-1+2c}\left(\sum_{j\in[n]\backslash \{i\}}\frac{1}{|\lambda_j-\lambda_i|}\right)^{k}\\
&\leq&k(k!)2^k(\log n)^{8k}n^{-1+2c+kc}.
\end{eqnarray*}
 Here we use the three assumptions of $\omegam$. Notice that it is key that $c<1/2$, which means we cannot consider all bulk eigenvectors at once. In order to use this on $\pphi(u)\pphi(v)$, we write $\pphi(u)\pphi(v)=\frac12\left(\langle \overrightarrow{1_u}+\overrightarrow{1_v},\pphi\rangle^2-\langle \overrightarrow{1_u},\pphi\rangle^2-\langle \overrightarrow{1_v},\pphi\rangle^2\right)$. Therefore we orthogonalize each of these to $\overrightarrow{1_{[n]}}$ with the three vectors
\begin{eqnarray*}
w_1&=&\frac{1}{\sqrt{2+\frac{4}{n-2}}}(\overrightarrow{1_u}+\overrightarrow{1_v}-\frac{2}{n-2}\overrightarrow{1}_{[n]\backslash \{u,v\}})\\
w_2&=&\frac{1}{\sqrt{1+\frac{1}{n-1}}}(\overrightarrow{1_u}-\frac1{{{n-1}}}\overrightarrow{1}_{[n]\backslash \{u\}})\\
w_3&=&\frac{1}{\sqrt{1+\frac{1}{n-1}}}(\overrightarrow{1_v}-\frac1{{{n-1}}}\overrightarrow{1}_{[n]\backslash \{v\}}).
\end{eqnarray*}
We denote the target vectors $\frac1{\sqrt{2}}(\overrightarrow{1_u}+\overrightarrow{1_v}),\overrightarrow{1_u},\overrightarrow{1_v}$ by $v_1,v_2,v_3$. We have for $i=1,2,3$,
\begin{eqnarray*}
\langle v_i,\pphi\rangle^2-\langle w_i,\pphi\rangle^2&=&\langle v_i+w_i ,\pphi\rangle\cdot \langle v_i-w_i ,\pphi\rangle\\
&=&\langle (v_i-w_i)+2w_i,\pphi\rangle\cdot\langle v_i-w_i ,\pphi\rangle.
\end{eqnarray*}
Therefore, if we define $y_i(\pphi)=\langle (v_i-w_i)+2w_i,\pphi\rangle\cdot\langle(v_i-w_i ),\pphi\rangle$, then 
\begin{equation}\label{eq:quaddecomp}
\pphi(u)\pphi(v)=\frac{1}{2}\left(2\langle w_1,\pphi\rangle^2-\langle w_2,\pphi\rangle^2-\langle w_3,\pphi\rangle^2\right)+\frac12(2y_1(\pphi)-y_2(\pphi)-y_3(\pphi)).
\end{equation}

Similar to before, we can write,
\begin{equation*}
k(k!)|\partial_{ab}^{(k)}y_i(\pphi)|=\left|\sum_{j\in ([n]\backslash \{i\})^k}\frac{\langle v_i-w_i ,\bm{\varphi}_{j_1}\rangle\langle (v_i-w_i)+2w_i,\bm{\varphi}\rangle(\pphi_{j_k}^TV\pphi)\prod_{\ell\in [k-1]} (\bm{\varphi}_{j_\ell}^TV\bm{\varphi}_{j_{\ell+1}})}{\prod_{\ell\in [k]}(\lambda_{j_\ell}-\lambda_i)}\right|.
\end{equation*}
Under the assumption of $\omegam$, and the fact that $\|v_i-w_i\|=O(n^{-1/2}) $,
\[\langle (v_i-w_i)+2w_i,\pphi\rangle\cdot \langle(v_i-w_i),\pphi_{j_1}\rangle=O(n^{-1+c}).\]
Therefore, by the same argument as before,
\[
|\partial_{ab}^{(k)}y_i(\pphi)|\leq k(k!)2^k(\log n)^{8k}n^{-1+2c+kc}.
\]

By \eqref{eq:quaddecomp} and the previous derivative bounds,
\begin{eqnarray*}
|\psi^{\tilde A}_{uv}(\tilde M)-\pphi(u)\pphi(v)|&\leq &4\sum_{k=1}^\infty (\frac{2}{\sqrt{(p+q)n}})^{k}\sum_{x_1,\ldots, x_k\in[\binom s2]}\frac{1}{k_1!k_2!\cdots k_{\binom{s}{2}}!}\left| \left[\frac{\partial^k}{\partial x_1,\ldots x_k}\psi^{\tilde A}_{uv}\right](\tilde A_S)\right|\\
&\leq&4\binom s2 n^{-1+2c} \sum_{k=1}^\infty (\frac{2}{\sqrt{(p+q)n}})^{k}k2^k(\log n)^{8k} n^{ck}\sum_{x_1,\ldots ,x_k\in[\binom s2]}\frac{k!}{k_1!k_2!\cdots k_{\binom{s}{2}}!}\\
&\leq&4\binom s2 n^{-1+2c} \sum_{k=1}^\infty (\frac{2}{\sqrt{(p+q)n}})^{k}k2^k(\log n)^{8k} n^{ck}{\binom s2}^k
\end{eqnarray*}
as we are summing over all multinomials. 
Therefore, we have
\[
|\psi^{\tilde A}_{uv}(\tilde M)-\pphi(u)\pphi(v)|=O(n^{-3/2+3c})
\]

As $F$ is a polynomial, assuming $\omegam$, $|(\partial^{k}F)( A_S,\sqrt n\pphi(S)))|\leq \log^C n$ for some fixed constant $C$ and any direction. Therefore, assuming $\omegam$, we can once again expand according to the partial derivatives in the direction $\psi_{uv}$.
\begin{align*}
\left|F( A_S,\sqrt{n}\psi^{\tilde A}(\tilde M))-F( A_S,\sqrt{n}\pphi(S))\right|
&\leq \sum_{k=1}^{\deg(F)} \frac{\left|\partial^{(k)}(F( A_S,\sqrt{n}\pphi(S)))\right|n^{k/2} |\psi^{\tilde A}_{uv}(\tilde M)-\pphi(u)\pphi(v)|^k}{k!}\\ &=O(n^{-1+3c}).
\end{align*}
\end{proof}

\subsection{Part IV: Quantum Unique Ergodicity}
 In this section, we utilize the powerful quantum ergodicity theorem to reduce our problem to one on polynomials of Gaussians. To do this, we introduce some new notation. 
\begin{definition}We consider the following different expectations. 
  \begin{enumerate}
\item 
$\E_A$ refers to the expectation across the choice of the random adjacency matrix $A$.
\item
$\E_S$ refers to the {conditional expectation given a matrix $A$} over a randomly selected $s\times s$ principal submatrix of $A$, for $s$ constant.
\item
$\E_M$ refers to the {conditional} expectation given a matrix $A$ and a subset of vertices $S\subset V$, over replacing $A_S$ with $M\in \mathbf M$, according to the distribution of $\mathbf M$.
\item
${\E_{\mathbf g}}$ refers to the expectation over some function of $\mathbf g$. 
\end{enumerate}
\end{definition}

Here, we present a reduced form of the  quantum unique ergodicity theorem for eigenvectors.

{
\begin{lemma}\label{lem:qe}[\cite{bourgade2017eigenvector} Theorem 1.5]
Fix $\epsilon>0$ and a finite degree polynomial $F:\R\rightarrow \R$. There is a constant $\nu>0$ such that for any unit vector $w\perp \overrightarrow{1_{[n]}}$ and eigenvector $\pphi$ of $A$ of index $i\in [\epsilon n,(1-\epsilon) n]$, 
\[
|\E_{A}[F(n\langle \bm{\varphi},w\rangle^2)]-\E_{\mathbf g}(F( \mathbf{g}^2))|=O(n^{-\nu})
\]
where $\mathbf g^2$ is a vector of squared normalized independent Gaussians.
\end{lemma}
}

The following is an application of this lemma as vectors supported on a small subset $S$ are close to orthogonal to $\overrightarrow{1_{[n]}}$.
\begin{lemma}\label{lem:qeconcentration}[\cite{huang2020size} Lemma 2.1]
For any finite degree even polynomial $F:\R^s\rightarrow \R$, and finite set of vertices $S$, there is a constant $\nu_1>0$ such that 
\begin{equation}\label{eq:finiteqe}
|\E_{A}(F(\sqrt{n}\pphi(S)))-\E_{\mathbf g}(F( \mathbf{g}))|=O(n^{-\nu_1}).
\end{equation}
\end{lemma}

Here, $F$ does not take $A_S$ as an input. However, because of the near independence of $A_S$ and $\pphi(S)$ in $p_s$ shown in Lemma \ref{lem:decor}, we can translate this into a similar bound on $p_s({A}_S,\pphi(S))$.

\begin{lemma}\label{lem:partcheb}
   For any $C,\delta>0$ we consider $p_s$ defined in Lemma \ref{lem:defp} with parameters $C,\delta$. {For $i$ as previously defined}, there is a constant $\nu>0$ such that with probability $1-O(n^{-\nu})$, for $m\in \{1,2\}$, 
    \[
    |\E_S(p_s({A}_S,\pphi(S))^m)-\E_{\mathbf g}(\E_M(p_s({M},\mathbf{g})^m))|=O(n^{-\nu}).
    \]

\end{lemma}

\begin{proof}[Proof of Lemma \ref{lem:partcheb}]
Assume that $m=1$. The proof for $m=2$ is identical. We claim that in fact it is enough to show that there is some constant $\nu'$ such that with probability $1-O(n^{-\nu'})$,
\begin{equation}\label{eq:partcheb}
\E_S(p_s({A}_S,\pphi(S)))=\E_A\left[\E_M(p_S(M,\pphi(S)))\right]+O(n^{-\nu'}).
\end{equation}

To see why this is sufficient, assuming \eqref{eq:partcheb}, then by Lemma \ref{lem:qeconcentration} with $F=\E_M(p_s({M},\sqrt{n}\pphi(S)))$, with probability $1-O(n^{-\nu'\wedge \nu_1})$, 
\[
\E_S(p_s({A}_S,\pphi(S)))=\E_{\mathbf g}(\E_M(p_s({M},\mathbf{g})))+O(n^{-\nu'\wedge\nu_1})
\]
as desired.

 We will prove \eqref{eq:partcheb} through a Chebyshev inequality. Therefore we need to calculate the expectation and variance. First, for $M\in \mathbf M$, we define
\begin{equation}\label{eq:onepart}
X_M:=\E_S(p_s({A}_S,\pphi(S))\vone(A_S=M)).
\end{equation}

By Lemma \ref{lem:decor}, under $\omegam$, $p_s({A}_S,\pphi(S))=\E_{M_1}(p_s( A_S,\psi^{\tilde A}( M_1)))+O(n^{-1+3c})$, for $M_1$ distributed according to $\mathbf M$. Therefore,
\begin{eqnarray*}
\E_A\left(X_M\vone_{\omegam}\right)&=&\E_A(\left[\E_S(p_s(A_S,\pphi(S))1(A_S=M))\right]\vone_{\omegam})\\
&=&\pr(A_S=M)\left(\E_A\left(\E_S(\E_{M_1}(p_s(A_S,\psi^{\tilde A}( M_1)))\vone_{\omegam}\right)+O(n^{-1+3c})\right)\\
&=&\pr(A_S=M)\left(\E_A\left(\E_S(\E_{M_1}(p_s(M_1,\pphi(S))))\vone_{\omegam}\right)+O(n^{-1+3c})\right)\\
&=&\pr(A_S=M)\left(\E_A\left(\E_{M_1}(p_s(M_1,\pphi(S)))\vone_{\omegam}\right)+O(n^{-1+3c})\right)
\end{eqnarray*}
for any fixed $S$. Therefore, summing over all $M$ gives
\begin{equation}\label{eq:expectation}
\sum_{M\in \mathbf M}\E_A\left(X_M\vone_{\omegam}\right)=\E_A\left(\E_{M_1}(p_s(M_1,\pphi(S)))\vone_{\omegam}\right)+O(n^{-1+3c}).
\end{equation}

Next we will deal with the second moment. We once again shift $p_s$ according to Lemma \ref{lem:decor},
\begin{eqnarray*}
\E_A((\sum_{M\in \mathbf M}X_M)^2\vone_{\omegam})
&=&\sum_{M,M'\in \mathbf M}\E_A(\E_S[p_s({M},\pphi(S))1(A_S=M)]\E_S[p_s({M'},\pphi(S))1(A_S=M')]\vone_{\omegam})\\
\hspace*{\fill}&=&\sum_{M,M'\in \mathbf M}\E_A\left(\frac{1}{\binom{n}s^2}\sum_{S_1,S_2\in \binom{[n]}{s}}(p_s(M,\pphi(S_1))(p_s(M',\pphi(S_2))1(A_{S_1}=M,A_{S_2}=M'))\vone_{\omegam}\right)\\
&=&\sum_{M,M'\in \mathbf M}\E_A\Bigg{(}\frac{1}{\binom{n}s^2}\sum_{S_1,S_2\in \binom{[n]}{s}} \left(\E_{M_1,M_2}(p_s(M,\psi^{\tilde A}( M_1)) p_s(M',\psi^{\tilde A}( M_2)))+O(n^{-1+3c})\right)
\\&& \times\pr(A_{S_1}=M,A_{S_2}=M')\vone_{\omegam}\Bigg{)}
\end{eqnarray*}

We have
$\pr(A_{S_1}=M,A_{S_2}=M')=\pr(A_{S_1}=M)\pr(A_{S_2}=M')$ if $S_1\cap S_2=\emptyset$.  There are at most $\binom nss\binom n{s-1}$ sets $S_1,S_2$ such that $S_1\cap S_2\neq \emptyset$. By \eqref{eq:infnormbound}, assuming $\omegam$, $p_s({M},\pphi(S))=(\log n)^{O(1)}$. Therefore
\begin{multline*}
    \E_A((\sum_{M\in \mathbf M}X_M)^2\vone_{\omegam})=
   \sum_{M,M'\in \mathbf M}\ \E_A\Bigg{(}\frac{1}{\binom{n}s^2}\sum_{S_1,S_2\in \binom{[n]}{s}} \left(\E_{M_1,M_2}(p_s(M_1,\pphi( S_1)) p_s(M_2,\pphi( S_2)))\right)
\\ \times\pr(A_{S_1}=M)\pr(A_{S_2}=M')\vone_{\omegam}\Bigg{)}+O(n^{-1+3c}).
\end{multline*}
which is
\begin{equation}\label{eq:secondmoment}
\E_A\left[\E_M(p_S(M,\pphi(S)))\right]^2+O(n^{-1+3c}).
\end{equation}
Combining \eqref{eq:expectation} and \eqref{eq:secondmoment} gives a Chebyshev inequality,
\begin{equation}\label{eq:slasteqq} 
\pr(\left|\E_S(p_s({A}_S,\pphi(S)))-\E_A\left[\E_{S}(\E_M(p_S(M,\pphi(S))))\vone_{\omegam}\right]\right|\geq n^{-c})=O(n^{-((1+c)\wedge\gamma)}).
\end{equation}

We are almost done; we just need to remove the dependence on $\omegam$ on the right hand side of \eqref{eq:slasteqq}. Therefore, we use \eqref{eq:finiteqe} for  $F:=\E_M[p_s( M,\pphi(S))^2]$ to obtain that for some constant $\nu_2$,
\begin{eqnarray}\nonumber
|\E_A\left[\E_M(p_S(M,\pphi(S)))\overline\vone_{\omegam}\right]|&\leq& |\E_A\left[\E_M(p_S(M,\pphi(S)))^2\right]|^{1/2}\pr(\overline\vone_{\omegam})^{1/2}\\
\nonumber&\leq& |\E_A\left[\E_M(p_S(M,\pphi(S))^2)\right]|^{1/2}\pr(\overline\vone_{\omegam})^{1/2}\\
\nonumber&\leq& |\E_{\mathbf g}( \E_M(p_s({M},\mathbf{g})^{2}))|^{1/2}(1+O(n^{-\nu_2}))\pr(\overline\vone_{\omegam})^{1/2}\\
&=& O(n^{-\gamma/2})\label{eq:lasteqq}
\end{eqnarray}
where the second to last line follows from Lemma \ref{lem:qeconcentration} and the last line follows from \eqref{eq:ell2bound} and Lemma \ref{lem:delocalization}. \eqref{eq:partcheb} follows from combining \eqref{eq:slasteqq} and \eqref{eq:lasteqq}. 
\end{proof}
\subsection{Part V: Counting Domains}
We are now ready to think more directly about counting domains. We will show in this section, that with Gaussian inputs, the expectation is large. First, we show that $\E_{\mathbf g}(\E_A[f_s(A_S,\mathbf{g})])$ is large.

\begin{lemma}\label{lem:fexp}
    For any set $|S|=(p\wedge q)^{-k}$, 
    \[
    \E [\E_A(f_s(A_S,\mathbf{g}))]\geq 1- \exp(-(p\wedge q)^{-k^2(1/2-o_k(1))}).
    \]
\end{lemma}
\begin{proof}
Without loss of generality, assume that $p\geq q$. Fix some set $S$. Define $I_\alpha$ to be the event that the set of $k$ vertices $\alpha$ forms a $k$-clique nodal domain. We then define $X:=\sum_{\alpha \in \binom{[S]}{k}}I_\alpha $. In order to quantify the dependence across different $\alpha$, we define
\[
\Delta:=\E_M(X)+\sum_{\alpha\sim\beta}\E_M( I_\alpha I_\beta),
\]
where $\alpha\sim \beta$ if $\E_M(I_\alpha I_\beta)>\E_M(I_\alpha)\E_M(I_\beta)$.
By Janson's inequality, (\cite{janson2011random}, see \cite{alon2016probabilistic} Theorem 8.1.2)
\begin{equation}\label{eq:janson}
\pr(X=0)\leq \exp(-\E_M(X)^2/\Delta).
\end{equation}

In order to utilize this, we compute the overlap,
\begin{eqnarray*}
\sum_{\alpha\sim\beta}\sum_{\beta\sim \alpha}\E_M(I_\alpha I_\beta)&\leq &\sum_{\alpha\sim\beta}\sum_{r=2}^{k} \binom kr\binom {s-k}{k-r} q^{-\binom{r}{2}}\E_M(I_\alpha)^2\\
&\leq&2\sum_{\alpha\in \binom{S}{k}}\binom k2\binom {s-k}{k-2} q^{-1}\E_M(I_\alpha)^2\\
&=&o_k(1) \E_M(X)^2.
\end{eqnarray*}

Therefore by \eqref{eq:janson},
\begin{equation}\label{eq:probbound}
\pr(X=0)\leq \exp\left[-(1-o_k(1))\E_M(X)\right].
\end{equation}

In order to calculate $\E_M(X)$, consider a specific $\alpha\in \binom{[S]}{k}$. 
For a vertex pair $(u,v)\in \alpha$, $\pr(A_{uv}g_u g_v>0)\geq q$. Therefore, the probability that all edges in $\alpha$ are positive is at least $q^{\binom k2}$, 
giving
\[
\E_M(X)\geq q^{\binom k2}\binom sk.
\]
Plugging this into \eqref{eq:probbound} gives
\[\pr(X=0)\leq \exp\left[-(1-o_k(1))q^{\binom k2}\binom sk\right].\]

We chose $s=q^{-k}$ in \eqref{eq:defs}. By then using the approximation $\binom{m}{j}\geq \frac{m^j}{j^j}$, we have
\[\pr(X=0)\leq \exp(-q^{-k^2(1/2-o_k(1))})\]
and therefore
\[\E_{\mathbf g}(\E_M(f_s(\mathbf{g})))\geq 1- \exp(-q^{-k^2(1/2-o_k(1))}).\]
\end{proof}

Before we prove Theorem \ref{thm:mainavg}, we need one more eigenvector structure result.
\begin{lemma}[\cite{rudelson2016no} Theorem 1.5]\label{lem:nogaps}
There is some constant $c_2$ such that for any $\epsilon>0$, with probability $1-n^{-\omega(1)}$, every eigenvector of $A$ satisfies
\[
\|v(I)\|_2\geq (\epsilon \cdot e^{-c_2/\sqrt{\epsilon}})^7\|v\|_2.
\]
for all $I\subset [n], |I|\geq \epsilon n$.
\end{lemma}

Note that this implies that with probability $1-n^{-\omega(1)}$, for some constant $c>0$, 
\begin{equation}\label{eq:nogaps2}
\left|\left\{v:|\pphi(v)|\leq \sqrt{\frac{{\delta}}{{n}}} \right\}\right|\leq \frac{c}{(\log \frac 1\delta)^2}n.
\end{equation}
 
\begin{proof}[Proof of Theorem \ref{thm:mainavg}]
Throughout this proof, when we write $p_s(\cdot)$, namely with one input, we mean the function $p_s:\R^{s}\rightarrow \R$ defined as $p_s(x):=\E_M(p_s({M},x))$.

Given the statement of Lemma \ref{lem:fexp} concerning $\E_A(f_s(A_S,\mathbf{g}))$, we translate this to a bound on the polynomial $p_s$. For fixed $C,\delta>0$, define the event 
$F_{C,\delta}$ as the event that $\forall u,v\in S, g_{u}g_v\in [-{C},0]\cup [{\delta}, {C}]$. We then have the decomposition
\begin{eqnarray*}
\left|\E_{\mathbf g}(p_s(\mathbf{g}))-\E_{\mathbf g}(f_s(\mathbf{g}))\right|&\leq&\left|\E_{\mathbf g}(p_s(\mathbf{g})\vone_{F_{C,\delta}})-\E_{\mathbf g}(f_s(\mathbf{g})\vone_{F_{C,\delta}})\right|+ \left|\E_{\mathbf g}(p_s(\mathbf{g})\overline\vone_{F_{C,\delta}})\right|+\left|\E_{\mathbf g}(f_s(\mathbf{g})\overline\vone_{F_C,\delta})\right|.
\end{eqnarray*}
By \eqref{eq:pinf},
\[
\left|\E_{\mathbf g}(p_s(\mathbf{g})\vone_{F_{C,\delta}})-\E_{\mathbf g}(f_s(\mathbf{g})\vone_{F_{C,\delta}})\right|\leq \delta.
\]
For the next term, we use \eqref{eq:pinf2}. To bound the probability of $\overline F_{C,\delta}$, note the maximum of the PDF of the univariate standard normal is $\frac{1}{\sqrt{2\pi}}$, Therefore, by \eqref{eq:ell2bound},
\begin{equation*}\label{eq:pbound2}
\left|\E_{\mathbf g}(p_s(\mathbf{g})\overline \vone_{F_{C,\delta}})\right|\leq \E_{\mathbf g}(p_s(\mathbf{g})^2)^{1/2}\pr(\overline F_{C,\delta})^{1/2}\leq \sqrt 2(s\sqrt{\frac{\delta}{2\pi}}+2s^2e^{-C})^{1/2}.
\end{equation*}
The last error term is 
\[
\left|\E_{\mathbf g}(f_s(\mathbf{g})\overline\vone_{ F_{C,\delta}})\right|\leq s^2e^{-C   }.
\]

Combining this estimation with Lemma \ref{lem:partcheb} and Lemma \ref{lem:fexp} gives that, for sufficiently small $\delta$ and sufficiently large $C$, there is some $\nu>0$ such that with probability $1-O(n^{-\nu})$, 
\[\E_S(p_s({A}_S,\pphi(S)))\geq 1-\delta-\exp(-q^{-k^2}(1/2-o_k(1))).\]

Now, we want to limit the contribution of the bad sets in $\pphi(S)$. Here, we set $F_{C,\delta}(\pphi_S)$ to be the event that for each pair of vertices $u,v\in S$, $n\pphi(u)\pphi(v)\in ([-C,-\delta]\cup [\delta,C])$. To give an lower bound on this probability, we use \eqref{eq:nogaps2} and the trivial bound that the number of vertices of value at least $\sqrt{C/n}$ is at most $n/C$. By Lemma \ref{lem:partcheb},
\begin{eqnarray*}
|\E_S(p_s({A}_S,\pphi(S))\overline \vone_{F_{C,\delta}(\pphi_S)})|&\leq &|\E_S(p_s({A}_S,\pphi(S))^2)|^{1/2}\pr(\overline F_{C,\delta}(\pphi_S))^{1/2}\\
&\leq&(1+O(n^{-\nu}))|\E_{\mathbf g}(\E_M(p_s({M},\mathbf{g})^2)|^{1/2}\pr(\overline F_{C,\delta}(\pphi_S))^{1/2}\\
&\leq&(1+O(n^{-\nu}))(1+\delta)(\frac{c}{(\log \frac1\delta)^2}+\frac1{C})^{1/2}.
\end{eqnarray*}
for some constant $c$.

Therefore, for sufficiently large $C$, 
\[
|\E_S(p_s({A}_S,\pphi(S))\vone_{F_{C,\delta}})|\geq 1-\epsilon
\]
for 
\begin{equation}\label{eq:epdef}
\epsilon:= \delta+\exp(-q^{-k^2(1/2-o_k(1))})+\frac{c_2}{(\log \frac1\delta)}
\end{equation}
where $c_2$ is some constant. 

We set $\delta$ such that the second term is the dominant term. Specifically,
\[
\delta\leq \exp\left[-\exp(q^{-k^2})\right]
\]
suffices.

We now interpret $p_s$ on the set $F_{C,\delta}$. By \ref{eq:pinf}, if $S$ has a $k$-clique domain, then $p_s\leq 1+\delta$, and if $S$ does not, $p_s\leq \delta$. Therefore, the fraction of $S$ that contains a clique domain is at least $\zeta$, where $\zeta$ satisfies $(1+\delta)\zeta+(1-\zeta)\delta=1-\epsilon$. This gives $\zeta=1-\delta-\epsilon\geq 1-2\epsilon$. Therefore, in our graph, at least $(1-2\epsilon)\binom ns$ of all sets of size $s$ contain a $k$ clique nodal domain. 

In order to show that we can partition our graph into $2n/k$ nodal domains, we proceed greedily. We start with our entire vertex set $[n]$. We arbitrarily select an $s$-set $S$ in our vertex set that contains a clique domain. We remove an arbitrary $k$-clique domain inside $S$ from the vertex set. We repeat this process with our new vertex set, and continue until there are no sets $s$ that contain a clique domain.

As at least $(1-2\epsilon)\binom{n}{s}$ $s$-sets contain a $k$-clique domain, this continues until there are at most $r$ vertices left, for
\[
\binom{r}{s}\leq 2\epsilon \binom{n}{s}.
\]
Once again using the approximation $\frac{m^j}{j^j}\leq \binom{m}{j}\leq \frac{m^je^j}{j^j}$, we have that
\begin{equation}\label{eq:rbound}
r^s\leq 2\epsilon n^se^s.
\end{equation}
By our definition of $\epsilon$ in \eqref{eq:epdef} and $s$ in \eqref{eq:defs},
\[
r\leq n\exp(-q^{-k^2(1/2-o_k(1)}/s+1)\leq n\exp(-q^{-k^2(1/2-o_k(1)+k}+1).
\]
Therefore, by this process, for large $k$ we have at most
\[
\frac nk+ne\epsilon^{1/s}\leq 2n/k
\]
nodal domains. 

\end{proof}

\section*{Acknowledgements}
The authors would like to thank Lior Alon and Sergey Lototsky for helpful discussions. We also thank Gregory Berkolaiko and Nikhil Srivastava for comments on an early version of this manuscript.

{ \small 
	\bibliographystyle{plain}
	\bibliography{main} }

\appendix
\section{Path nodal domains of signed Erd\H{o}s-R\'enyi graphs}\label{sec:averagepath}

Consider an Erd\H{o}s-R\'enyi signed graph $G(n,p,q)$, with $0<p,q<1$ fixed, and its associated adjacency matrix $A$, e.g., $A$ is a $n \times n$ symmetric matrix with diagonal entries equal to zero and i.i.d. off-diagonal pairs, each equal to $+1$ with probability $p$, $-1$ with probability $q$, and $0$ with probability $1-p-q$. We prove the following:

\begin{prop}\label{claim:1path}
Let $0<p,q<1$ be constant, and $A$ be the adjacency matrix of an Erd\H{o}s-R\'enyi signed graph $G(n,p,q)$. Then, with probability $1-n^{-\omega(1)}$, $\kappa(G^{>}_{\bm{\varphi}})=1$ for all eigenvectors $\bm{\varphi}$ of $A$.
\end{prop}

As mentioned in Section \ref{sec:intro}, in the study of nodal counts of adjacency matrices, it is standard convention to consider domains where $A_{ij} \bm{\varphi}(i)\bm{\varphi}(j)>0$, as the negative of the adjacency matrix is a generalized Laplacian,  which we follow below. Given that both $p$ and $q$ are constant and the result is independent of eigenvalue indexing, this convention has no impact on the result (e.g., the same result holds for $\kappa(G^{<}_{\bm{\varphi}})$).

We first recall a number of known results from random matrix theory.

\begin{lemma}\label{lm:results} Let $A$ be the adjacency matrix of $G(n,p,q)$, where $0<p,q<1$ are fixed. The following are true.
\begin{enumerate}
\item(Nonzero vector \cite{nguyen2017random})
With probability $1-n^{-\omega(1)}$, every eigenvector $\pphi$ is nonzero
\item(Combination of Lemma \ref{lem:nogaps} and \eqref{eq:infnormbound})
For any fixed $c$, with probability $1-n^{\omega(1)}$, every unit eigenvector $\bm{\varphi}$ is such that any set of vertices of size $|S|\geq c n$ satisfies
\begin{equation}\label{eq:qe1}
\sum_{u\in S}|\bm{\varphi}(u)|\geq \frac{\sqrt n}{\log^5 n
}.
\end{equation} 
\item(Bai Yin Theorem \cite{bai1988necessary}, see \cite[Theorem 4.4.5]{vershynin2018high})
With probability $1-n^{\omega(1)}$, $\|A - (p-q) \mathbf{1}\mathbf{1}^T \| \le \big(2 + o_n(1)\big)\sqrt{(p+q)n}$.
\end{enumerate}
\end{lemma}

To prove the desired result, we start with two weaker statements.

\begin{lemma}
With probability $1- n^{-\omega(1)}$, every eigenvector $\bm{\varphi}$ of the adjacency matrix has $\kappa(G^{>}_{\pphi})\leq 3\log_{\frac{1}{1-p \vee q}} n$. 
\end{lemma}
\begin{proof}
Suppose $\kappa(G^>_{\bm{\varphi}}) \geq k $. We denote by $S$ an arbitrary set of vertices with exactly one vertex taken from each connected component of $G^>_{\bm{\varphi}}$. The induced subgraph on $S$ cannot contain an edge $(u,v)$ satisfying $A_{uv}\pphi(u)\pphi(v) < 0$. We now ``re-sign'' the vertices, by multiplying $A$ on the left and the right by a diagonal matrix $D$, which has diagonal entry $1$ in all entries that are not a vertex in $S$, and $\sgn(\pphi(u))$ for vertices in $S$. According to this signing, $(DAD)_{i,j} 
\le 0$, for all vertices $u,v \in S$.

For fixed $D$ and random $A$, the probability that $(DAD)_{u,v}$ has the correct sign is at most $(1-p\vee q)$. Therefore, by union bounding over all possible $D$ and sets of vertices $S$, we have the probability of $\kappa(G^>_{\bm{\varphi}}) \geq k $ 
is at most
\[
\binom nk2^k(1-p\vee q)^{\binom k2}
\]
 This probability is $n^{-\omega(1)}$ for $k\geq 3\log_{\frac{1}{1-p \vee q}} n$.
\end{proof}

\begin{lemma}
With probability $1-n^{-\omega(1)}$, the second largest connected component of $G^>_{\bm{\varphi}}$ is at most one vertex.  
\end{lemma}

\begin{proof}
 Denote by $U$ the vertex set of the largest connected component of $G^>_{\bm{\varphi}}$, which, by the previous lemma, must be of size at least $n/(3\log_{\frac{1}{1-p \vee q}} n)$. Any connected component $V\neq U$ of size $k$ must have all edges $(u,v)\subset U\times V$ satisfy $A_{uv}\pphi(u)\pphi(v)\leq 0$. 

We now consider a different re-signing, by multiplying $A$ on the left and the right by a diagonal matrix $D$, which has diagonal entry $1$ in all entries that are not in $V$, and $\sgn(\pphi(v))$ for $v\in V$. For $u\in U$ and $v\in V$, we have $\sgn(A_{uv}\pphi(u)\pphi(v))=\sgn((DAD)_{uv}\pphi(u))\leq 0$.

Assume that $|V|\geq 2$. For $u\in U$ and $v_1,v_2\in V$, if $(u,v_1),(u,v_2)$ are edges, then 
\begin{equation}\label{eq:resign}
\sgn((DAD)_{uv_1})=\sgn((DAD)_{uv_2}).
\end{equation}
With probability $1-n^{\omega(1)}$, the number of vertices in $U$ in the shared neighborhood of $v_1,v_2$ is $(1+o_n(1))(p+q)^2|U|$. 
For any fixed signing $D$, the probability of $(u,v_1),(u,v_2)$ having the same sign is at most $p^2+q^2$. Set $k=|V|$. Considering $|U|\geq n/(3\log_{\frac{1}{1-p \vee q}}n)$, if we union bound over all $\binom nk$ possible sets $V$ and $2^k$ signings,
the probability of there being a signing $D$ that satisfies \eqref{eq:resign} is at most
\[
\binom nk 2^{k}(p^2+q^2)^{(p+q)^2n/(3\log_{1/(1-p\vee q)} n)\binom{k}2}.
\]
For $k\geq2$, this probability is $n^{-\omega(1)}$.
\end{proof}

We are now prepared to prove Proposition \ref{claim:1path}.

\begin{proof}[Proof of Proposition \ref{claim:1path}]
 We break our analysis into two cases, treating the spectral radius separately from the rest of the spectrum. 

Suppose $\pphi$ is a unit eigenvector corresponding to the spectral radius of $A$. By Fact (3) of Lemma \ref{lm:results}, with probability $1- n^{-\omega(1)}$, $A$ has one eigenvalue equal to $\big( 1 + o_n(1) \big)(p-q)n$ and all other eigenvalues of magnitude at most $\big(2 + o_n(1)\big)\sqrt{(p+q)n}$. Furthermore, with probability $1- n^{-\omega(1)}$, 
$$\overrightarrow{1_{[n]}}^T A \, \overrightarrow{1_{[n]}}= \sum_{i ,j = 1}^{n} A_{ij} = \big( 1 + o_n(1) \big) (p-q)n(n-1).$$
Now, let us consider the representation of $\overrightarrow{1_{[n]}}$ in the eigenbasis of $A$. We have
$$ \big|\overrightarrow{1_{[n]}}^T A \, \overrightarrow{1_{[n]}}\big| \le \big( 1 + o_n(1) \big)|p-q|n \, (\overrightarrow{1_{[n]}}^T \pphi_i)^2 + \big(2 + o_n(1)\big)\sqrt{(p+q)n}\, \big( n - (\overrightarrow{1_{[n]}}^T \pphi_i)^2\big).$$
Therefore, with probability $1- n^{-\omega(1)}$, $(\overrightarrow{1_{[n]}}^T \pphi)^2 =\big(1-o_n(1) \big) n$, and so $$\|\pphi- \sgn(\overrightarrow{1_{[n]}}^T \pphi)\overrightarrow{1_{[n]}}/\sqrt{n}\|=o_n(1).$$ This implies that $\pphi$ has constant sign, say $\pphi(u) >0$, on $\big(1-o_n(1) \big) n$ of its vertices. Now, consider an isolated vertex $v$ in $G^>_{\bm{\varphi}}$. With probability $1-n^{-\omega(1)}$, $v$ has edges to $(1+o_n(1))pn$ different vertices $u$ satisfying $\pphi(u) >0$, $A_{u,v}=+1$ and $A_{u,v}\pphi(u)\pphi(v)<0$, and edges to $(1+o_n(1))qn$ different vertices $w$ satisfying $\pphi(w) >0$, $A_{w,v}=-1$ and $A_{w,v}\pphi(w)\pphi(v)<0$, a contradiction. Therefore, $\kappa(G^{>}_{\pphi})>1$ with probability $n^{-\omega(1)}$.

Now, consider an arbitrary unit eigenvector $\pphi$ corresponding to an eigenvalue in the rest of the spectrum. Suppose there is a $v$ that is an isolated vertex in $G^>_{\bm{\varphi}}$. With probability $1-n^{-\omega(1)}$, $v$ has edges to $(1+o_n(1))(p+q)n$ different vertices, $u$, all of which satisfy $A_{u,v}\pphi(u)\pphi(v)<0$. The eigenvector equation at $v$ gives
\begin{equation*}
\lambda \pphi(v)=\sum_{u\sim v}A_{uv}\pphi(u).
\end{equation*}
Therefore, by the eigenvector equation at $v$ and noting that $A_{u,v}\pphi(u)\pphi(v)<0$ for $u \sim v$, we have
\begin{equation}\label{eq:eigeq}
|\lambda \bm{\varphi}(v)|= \Big|\sum_{u\sim v}A_{uv}\pphi(u) \Big| = \sum_{u\sim v} |\bm{\varphi}(u)|.
\end{equation}
Plugging this into \eqref{eq:qe1},
\[
|\lambda\bm{\varphi}(v)| \geq  \frac{\sqrt{n}}{\log^5 n}. 
\]
However, considering the infinity norm bound on the eigenvector from \eqref{eq:infnormbound} and the Bai Yin theorem (Fact (3) of Lemma \ref{lm:results}), this happens with probability at most $n^{-\omega(1)}$.
\end{proof}

\section{An illustrative example}\label{sec:decompositionexample} Consider the $16$ by $16$ symmetric matrix $M$ with
\begin{align*}
    M_{ii} &= \begin{cases} -1 & \text{ for } i = 1,2,5,10,11,12 \\ 0 & \text{ otherwise } \end{cases} \;,\\
    M_{ij} &= \begin{cases} 1 & \text{ for } \{i,j\}= \{6,9\},\{15,16\}  \\ -1 & \text{ for } \{i,j\}= \{1,2\},\{1,5\},\{2,5\},\{3,8\},\{4,8\},\{5,6\},\{5,9\},\{6,7\}, \{6,11\}, \\ &   \qquad \qquad \qquad \{8,11\}, \{9,10\}, \{9,13\},\{9,14\},\{10,11\},\{10,12\},\{10,15\},\\ &   \qquad \qquad \qquad \{10,16\},\{11,12\},\{11,15\},\{11,16\},\{12,15\},\{12,16\}  \\ 0 & \text{ otherwise } \end{cases}
\end{align*}
for $i\ne j$. See Figure \ref{fig:graph} for the signed sparsity graph $G = ([16],E,\sigma)$ of $M$. Consider the eigenvalue $\lambda = 0$ of $M$, with index $k=7$, multiplicity $r = 6$, and corresponding eigenspace
$$E_\lambda = \left\{ \pphi \, \left|   \; \begin{array}{rl} 0 &=\, \pphi(6) = \pphi(8) = \pphi(9) = \pphi(15) = \pphi(16) \\
 &=\, \pphi(1)+  \pphi(2) +\pphi(5)  \\ 
 &=\,   \pphi(5) + \pphi(7) + \pphi(11) \\ 
 &=\, \pphi(5) + \pphi(10) + \pphi(13) + \pphi(14)  \\ 
 &=\,  \pphi(10) + \pphi(11)  + \pphi(12)   \\ &=\, \pphi(3) + \pphi(4) + \pphi(11)   \end{array} \right.\right\}$$
We have
$X_1 = \{6,9\}$, $X_2 = \{15,16\}$, $X_3 = \{8\}$, $Y_1 = \{1,2,5\}$, $Y_2 = \{ 7\}$, $Y_3 = \{ 13\}$, $Y_4 = \{14\}$, $Y_5 = \{10,11,12\}$, $Y_6 = \{3\}$, and $Y_7 = \{4\}$. The corresponding bipartite graph $H$ on vertices $X= \{x_1,x_2,x_3\}$ and $Y=\{y_1,y_2,y_3,y_4,y_5,y_6,y_7\}$ is in Figure \ref{fig:graph}.

Consider the following non-vanishing orthogonal\footnote{For simplicity in this example, we use an orthogonal basis rather than an orthonormal one. However, as a result, Equation \ref{eqn:innerprod} does not apply, and the norms of eigenvectors must be taken into account.} bases for the orthogonal projections of $E_\lambda$ onto $Y_1$, $Y_2$, $Y_3$, $Y_4$, $Y_5$, $Y_6$ and $Y_7$:
\begin{align*}
&\big\{\bm{\psi}^{(1)}_1\big|_{Y_1}, \bm{\psi}^{(1)}_2\big|_{Y_1}\big\} = \big\{\bm{\psi}^{(5)}_1\big|_{Y_5}, \bm{\psi}^{(5)}_2\big|_{Y_5}\big\} = \left\{ \begin{pmatrix} \;1 \\ \;2 \\ -3 \end{pmatrix} ,  \begin{pmatrix} -5 \\ \;4 \\ \;1 \end{pmatrix} \right\}, \\
 & \big\{\bm{\psi}^{(2)}_1\big|_{Y_2}\big\} =  \big\{\bm{\psi}^{(3)}_1\big|_{Y_3}\big\} =  \big\{\bm{\psi}^{(4)}_1\big|_{Y_4}\big\} =  \big\{\bm{\psi}^{(6)}_1\big|_{Y_6}\big\} =  \big\{\bm{\psi}^{(7)}_1\big|_{Y_7}\big\} = \big\{(1) \big\}, 
\end{align*}
and let us consider a vector in their span:
$$ \bm{\varphi} = \alpha_1^1\bm{\psi}^{(1)}_1 + \alpha_2^1\bm{\psi}^{(1)}_2 +  \alpha_1^2\bm{\psi}^{(2)}_1
+  \alpha_1^3\bm{\psi}^{(3)}_1 +  \alpha_1^4\bm{\psi}^{(4)}_1 +  \alpha_1^5\bm{\psi}^{(5)}_1 +  \alpha_2^5\bm{\psi}^{(5)}_2 +  \alpha_1^6\bm{\psi}^{(6)}_1 +  \alpha_1^7\bm{\psi}^{(7)}_1.$$
The vector $\bm{\varphi} \in E_\lambda$ if and only if  
$$\pphi(5)+  \pphi(7) +\pphi(11) =   \pphi(5) + \pphi(10) + \pphi(13)  + \pphi(14)= \pphi(3) + \pphi(4)  + \pphi(11) = 0,$$
which, in terms of $\alpha$'s, produces our $\gamma = 3$ homogeneous equations:
\begin{align*}
   -3\alpha_1^1 +\alpha_2^1+\alpha_1^2 +2\alpha_1^5 +4\alpha_2^5  &=0, \\
   -3\alpha_1^1 +\alpha_2^1+\alpha_1^3+\alpha_1^4 +\alpha_1^5 -5 \alpha_2^5 &=0, \\
    2\alpha_1^5 +4 \alpha_2^5 + \alpha_1^6 + \alpha_1^7 &=0.
\end{align*}
Performing Gaussian elimination, and solving for the pivots, we obtain:
\begin{align*}
    \alpha_1^7 &= - \alpha_1^6 +\alpha_1^2  + \alpha_2^1 -3 \alpha_1^1\\
     \alpha_2^5 &= (1/7)\alpha_1^4+  (1/7)\alpha_1^3 - (1/14) \alpha_1^2 + (1/14)\alpha^1_2 -(3/14) \alpha_1^1, \\
      \alpha_1^5 &=  -(2/7) \alpha_1^{4}-(2/7)\alpha_1^{3} -(5/14)\alpha_1^2 - (9/14) \alpha_2^1 + (27/14) \alpha_1^1,
\end{align*}
so $\Sigma = \{(\eta_1,\sigma_1),(\eta_2,\sigma_2),(\eta_3,\sigma_3)\} =\{(5,1),(5,2),(7,1)\}$, $\hat Y = \{y_1,y_2,y_3,y_4,y_6\}$, $q-\gamma = 7-3=4$, and $j_i =i$ for $i =1,...,q-\gamma$. Each $\pphi(i)$, $i \not \in i_0(\lambda)$, is a linear function of $\big\{\{\alpha_\sigma^{j}\}_{\sigma = 1,...,r_j}^{j = 1,...,q} \, \backslash \, \{\alpha_{\sigma_\ell}^{\eta_\ell} \}_{\ell=1}^\gamma \big\}$. For $i \in Y_j$ for some $y_j \in \hat Y$, the linear function is obvious. We write the explicit function for each $i \in Y_5 \cup Y_7$ below:
\begin{align*}
\pphi(4) \; &= \alpha^7_1 \qquad \quad \; \; \, = - \alpha_1^6 +\alpha_1^2  + \alpha_2^1 -3 \alpha_1^1 ,\\
\pphi(10) &=  \alpha_1^5 - 5\alpha_2^5 \; \; \, \;= -\alpha^4_1 - \alpha^3_1 - \alpha^1_2 + 3 \alpha^1_1,\\
\pphi(11) &=  2 \alpha_1^5 + 4\alpha_2^5 \; \; = - \alpha^2_1 - \alpha_2^1 +3 \alpha_1^1,\\
\pphi(12) &=  -3 \alpha_1^5 + \alpha_2^5 \; = \alpha^4_1+\alpha^3_1+\alpha^2_1+ 2\alpha^1_2 - 6 \alpha_1^1.
\end{align*}
We are now prepared to produce an eigenbasis $\{\pphi_1,\pphi_2,\pphi_3,\pphi_4,\pphi_5,\pphi_6\}$ and corresponding signings $\{\bm{\varepsilon}_1,\bm{\varepsilon}_2,\bm{\varepsilon}_3,\bm{\varepsilon}_4,\bm{\varepsilon}_5,\bm{\varepsilon}_6\}$ of $E_\lambda$ satisfying Theorem \ref{thm:basis_bound}. 

For $i_0(\lambda)$, we set $\bm{\varepsilon}_s(6) = \bm{\varepsilon}_s(9) = \sgn \big(\pphi_s(5)\big)$, $\bm{\varepsilon}_s(15) = \bm{\varepsilon}_s(16) = \sgn \big(\pphi_s(10)\big)$, and $\bm{\varepsilon}_s(8) = \sgn \big(\pphi_s(11)\big)$, $s = 1,2,3,4,5,6$. We need only consider $\pphi_s$, $s<q-\gamma$; an arbitrary non-vanishing extension suffices for $\pphi_4,\pphi_5,\pphi_6$. We have
\begin{align*}
    &\pphi_1: \, \Pi_E^1 = \{ {}_1\alpha_1^5,{}_1\alpha_2^5,{}_1\alpha_1^7\}, &\Pi_S^1& = \{ {}_1\alpha_1^2,{}_1\alpha_1^3,{}_1\alpha_1^4\},  &\Pi_O^1& = \emptyset, &\Pi_F^1& = \{ {}_1\alpha_1^1,{}_1\alpha_2^1,{}_1\alpha_1^6\}, \\
    &\pphi_2: \, \Pi_E^2 = \{ {}_2\alpha_1^5,{}_2\alpha_2^5,{}_2\alpha_1^7\},  &\Pi_S^2& = \{ {}_2\alpha_1^2,{}_2\alpha_1^3\},  &\Pi_O^2& = \{ {}_2\alpha_1^4\},  &\Pi_F^2& = \{{}_2\alpha_1^1,{}_2\alpha_2^1,{}_2\alpha_1^6 \},  \\
    &\pphi_3: \, \Pi_E^3 = \{ {}_3\alpha_1^5,{}_3\alpha_2^5,{}_3\alpha_1^7\},  &\Pi_S^3& = \{ {}_3\alpha_1^2\},  &\Pi_O^3& = \{{}_3\alpha_1^3,{}_3\alpha_1^4 \},  &\Pi_F^3& = \{{}_3\alpha_1^1,{}_3\alpha_2^1,{}_3\alpha_1^6 \}.
\end{align*}

We begin with $\pphi_1$, and follow the proof of Claim \ref{claim:bound}. We set ${}_1\alpha_1^1 = 0$. The variable ${}_1\alpha_2^1$ cannot be zero, but is otherwise unconstrained. We set ${}_1\alpha_2^1 = 1$, and so $\bm{\varepsilon}_1(6) = \bm{\varepsilon}_1(9) = \sgn(\pphi_1(5)) = \sgn(1)=+1$. The variable ${}_1\alpha_1^2$ cannot be zero, and must satisfy $$\sgn({}_1\alpha_1^2 \bm{\psi}_1^{(2)}(7)) = \sgn({}_1\alpha_1^2 )= \bm{\varepsilon}_1(6)  = +1 \quad  \text{ and } \quad \pphi_1(11)  = - {}_1\alpha^2_1 - 1 \ne 0.$$ We set ${}_1\alpha_1^2 = 1$, implying that $\sgn(\pphi_1(11)) = -1$, and so ${\bm \varepsilon}_1(8) = -1$. The variable ${}_1\alpha_1^3$ cannot be zero, and must satisfy $\sgn({}_1\alpha_1^3 \bm{\psi}_1^{(3)}(13)) = \sgn({}_1\alpha_1^3 )= \bm{\varepsilon}_1(9)  = +1$. We set ${}_1\alpha_1^3 =  1$. As in the previous two cases, ${}_1\alpha_1^4$ must be positive, but must also satisfy
$$\pphi_1(10)  = -{}_1\alpha^4_1 - 1 - 1 \ne 0 \quad \text{and} \quad 
\pphi_1(12)  = {}_1\alpha^4_1+1+1+ 2\ne 0.$$
We set ${}_1\alpha_1^4 = 1$, implying that $\sgn(\pphi_1(10) ) = -1$, and so 
$\bm{\varepsilon}_1(15) =\bm{\varepsilon}_1(16) =-1$. All that remains is to set ${}_1\alpha_1^6$. It must satisfy  $\pphi_1(4) = - {}_1\alpha_1^6 +1  +1\ne 0$. We set ${}_1\alpha_1^6 = 1$, giving
\begin{align*}
\pphi_1 &= (-5 ,4 ,1 ,1 ,1 ,0 ,1 ,0 ,0 ,-3 , -2,5 , 1,1 ,0 ,0 )^T,\\ 
{\bm \varepsilon}_1 &= ( -1,+1 ,+1 ,+1 ,+1 ,+1 ,+1 , -1, +1,-1 ,-1 ,+1 ,+1 ,+1 , -1, -1)^T.
\end{align*}
Next, we consider $\pphi_2$, and follow the proof of Claim \ref{claim:orth}. Here, we are only concerned with half-space and orthogonality conditions, as vanishing entries can be addressed using $\pphi_1$ and an appropriate Givens rotation (see Claim \ref{claim:main}). We set ${}_2\alpha_1^1 = 0$. The variable ${}_2\alpha^1_2$ cannot be zero, but is otherwise unconstrained. We set ${}_2\alpha^1_2 = 1$, and so $\bm{\varepsilon}_2(6) = \bm{\varepsilon}_2(9) = \sgn(\pphi_2(5)) = \sgn(1)=+1$. The variable ${}_2\alpha_1^2$ must satisfy $\sgn({}_2\alpha_1^2 \bm{\psi}_2^{(2)}(7)) = \sgn({}_2\alpha_1^2 )= \bm{\varepsilon}_2(6)  = +1$. We set ${}_2\alpha_1^2 = 1$. The variable ${}_2\alpha_1^3$ must satisfy $\sgn({}_2\alpha_1^3 \bm{\psi}_2^{(3)}(13)) = \sgn({}_2\alpha_1^3 )= \bm{\varepsilon}_2(9)  = +1$. We set ${}_2\alpha_1^3 =  1$. We also set ${}_2\alpha_1^6 =  1$. What remains is to set ${}_2\alpha_1^4$ so that 
\begin{align*}
    \langle \pphi_1,\pphi_2\rangle&= \langle \pphi_1, (-5 ,4 ,1 ,1 ,1 ,0 ,1 ,0 ,0 ,-{}_2 \alpha_1^4 -2 , -2,{}_2 \alpha_1^4+4 , 1,{}_2 \alpha_1^4 ,0 ,0 )^T \rangle  \\ &= 9 {}_2 \alpha_1^4 + 76 = 0.
\end{align*}
We set ${}_2\alpha_1^4=-76/9$, giving
\begin{align*}
\pphi_2 &= (-5 ,4 ,1 ,1 ,1 ,0 ,1 ,0 ,0 ,58/9 , -2,-40/9 , 1,-76/9 ,0 ,0 )^T,\\ 
{\bm \varepsilon}_2 &= ( -1,+1 ,+1 ,+1 ,+1 ,+1 ,+1 , -1, +1,+1 ,-1 ,-1 ,+1 ,-1 , +1, +1)^T.
\end{align*}
Finally, we consider $\pphi_3$. We set ${}_3\alpha_1^1 = 0$ and ${}_3\alpha_2^1 = 1$, implying that ${\bm \varepsilon}_3(6) = {\bm \varepsilon}_3(9) = +1$, and so ${}_3 \alpha_1^2$ must be positive. We set ${}_3 \alpha_1^2 = 1$, and also set ${}_3 \alpha_1^6 = 1$. We set ${}_3 \alpha_1^3$ and ${}_3 \alpha_1^4$ so that
\begin{align*}
    \langle \pphi_1,\pphi_3\rangle&= \langle \pphi_1, (-5 ,4 ,1 ,1 ,1 ,0 ,1 ,0 ,0 ,-{}_3 \alpha_1^3-{}_3 \alpha_1^4 -1 , -2,{}_3 \alpha_1^3+{}_3 \alpha_1^4+3 , {}_3 \alpha_1^3,{}_3 \alpha_1^4 ,0 ,0 )^T \rangle  \\ &= 9  {}_3 \alpha_1^3+ 9 {}_3 \alpha_1^4+ 67= 0, \\
    \langle \pphi_2,\pphi_3\rangle&= \langle \pphi_2, (-5 ,4 ,1 ,1 ,1 ,0 ,1 ,0 ,0 ,-{}_3 \alpha_1^3-{}_3 \alpha_1^4 -1 , -2,{}_3 \alpha_1^3+{}_3 \alpha_1^4+3 , {}_3 \alpha_1^3,{}_3 \alpha_1^4 ,0 ,0 )^T \rangle  \\ &=  -(89/9)  {}_3 \alpha_1^3 - (58/3) {}_3 \alpha_1^4+ (263/9) = 0.
\end{align*}
We set ${}_3 \alpha_1^3 = 55/3$ and ${}_3 \alpha_1^4 = -98/9$, giving
\begin{align*}
\pphi_3 &= (-5 ,4 ,1 ,1 ,1 ,0 ,1 ,0 ,0 ,-76/9 , -2,94/9 , 55/3,-98/9 ,0 ,0 )^T,\\ 
{\bm \varepsilon}_3 &= ( -1,+1 ,+1 ,+1 ,+1 ,+1 ,+1 , -1, +1,-1 ,-1 ,+1 ,+1 ,-1 , -1, -1)^T.
\end{align*}
Below we provide a nodal decomposition for each ${\bm \varepsilon}_s$, $s = 1,2,3$, that satisfies the bound of Theorem \ref{thm:basis_bound}:
\begin{align*}
    {\bm \varepsilon}_1&: \; \{1\}, \, \{2,5,6,7\}, \,\{3\},\,\{4\},\, \{8,10,11,15\}, \, \{9,13,14\}, \, \{12\},\, \{16\}, \\
     {\bm \varepsilon}_2&: \; \{1\}, \, \{2,5,6,7\}, \,\{3\},\,\{4\},\, \{8,11,12\}, \, \{9,10,13,15\}, \, \{14\},\, \{16\}, \\
      {\bm \varepsilon}_3&: \; \{1\}, \, \{2,5,6,7\}, \,\{3\},\,\{4\},\, \{8,10,11,15\}, \, \{9,13\}, \, \{12\},\, \{14\}, \, \{16\}.
\end{align*}
The required bounds for ${\bm \varepsilon}_s$, $s>3$, hold for any non-vanishing vector $\pphi$ in $E_\lambda$ with $\bm \varepsilon$ defined as above.

\end{document}